\documentclass[12pt]{article}

\usepackage{caption}
\usepackage{subcaption}
\usepackage{amsthm,amsmath,amsfonts,algorithm,graphicx, braket, mathtools, pifont, cite, authblk}
\usepackage[utf8]{inputenc}
\usepackage{fullpage}

\usepackage{xcolor}

\usepackage{hyperref}

\newtheorem {conjecture} {Conjecture}[section]

\theoremstyle{definition}
\newtheorem{proposition}{Proposition}

\theoremstyle{plain}

\newcommand{\kb}[2]{\ket{#1}\bra{#2}}
\newcommand{\bk}[2]{ \braket{#1 | #2}}

\newcommand{\mbi}{\mathbb{I}}

\DeclareMathOperator{\Tr}{Tr}

\newcommand{\Z}{\mathcal{Z}}
\newcommand{\X}{\mathcal{X}}

\newcommand{\ixaxb}{I(X^A : X^B)}

\newcommand{\izab}{I(Z^A : B)}
\newcommand{\izazb}{I(Z^A : Z^B)}
\newcommand{\om}{\omega}

\newcommand{\iab}{I(A : B)}

\newcommand{\hza}{H(Z^A)}
\newcommand{\hzb}{H(Z^B)}
\newcommand{\hzazb}{H(Z^AZ^B)}

\newcommand{\hxa}{H(X^A)}
\newcommand{\hxb}{H(X^B)}
\newcommand{\hxaxb}{H(X^AX^B)}

\newcommand{\hma}{H(M_i^A)}
\newcommand{\hmb}{H(M_i^B)}
\newcommand{\hmamb}{H(M_i^AM_i^B)}

\newcommand{\obsth}{\frac{1}{\sqrt{3}}}

\newcommand{\imax}{I_{max}(M^A : M^B)}

\floatname{algorithm}{Protocol}

\title{On the CQC Conjecture: A sufficient condition and an extension}
\author{Hasan Iqbal\thanks{hasan.iqbal@uwyo.edu}}
\affil{Department of Electrical Engineering and Computer Science}
\affil{University of Wyoming}
\date{
	\today
}
\begin{document}
	
	\maketitle
	
	\begin{abstract}
The CQC conjecture by Schneeloch et al. (Physical Review A 90.6, 2014) asserts that the sum of classical mutual information between two parties obtained by measuring individual systems in two mutually unbiased bases cannot exceed their shared quantum mutual information. This remarkable conjecture, which deserves more work, still remains open. In this article, using a result by Coles and Piani (Physical Review A 89.2, 2014), we derive a sufficient condition for which the CQC conjecture holds. We show that this condition validates the conjecture for a wider class of states compared to the states that were mentioned in the original work. Furthermore, we extend this conjecture to a higher number of mutually unbiased bases and prime dimensions. We prove this extended CQC conjecture on isotropic states of arbitrary prime dimensions and simulate it extensively for dimensions 3 and 5 for random bipartite states, observing no contradiction. 
	\end{abstract}
	\section{Introduction}
	Entropic uncertainty (EU) relations \cite{Coles2017Feb, Wehner2010Feb} are widely used in different areas of quantum information theory. These equations express fundamental limits of obtainable information from quantum systems if one performs different measurements \cite{nielsen2001quantum}. Although the entropy-based formulations of these equations are more widely studied, mutual information based formulations are also quite useful, as they capture both individual and combined entropies in a multipartite system. Hall \cite{hall1995information, hall1997quantum} worked on the mutual information based formulations and posed interesting open questions. For example, Hall discussed the idea of information exclusion relation, which states that information gain from a measurement could only come at the expense of information gain in a measurement that is mutually unbiased to the first one \cite{hall1995information}. In that work, after proving an upper bound on mutual information sum when the measured system could be correlated with a classical system, he also posed the open question about a similar upper bound when an arbitrary number of mutually unbiased bases (MUB) \cite{durt2010mutually} are involved. These correlated systems, however, could be quantum too. As we have seen in the case of EU relation involving quantum memories formulated by Berta et al. \cite{berta2010uncertainty}, it is interesting to ask about upper bounds for information exclusion relation when the correlated system is quantum and different MUBS are involved. This is precisely what Schneeloch et al. asked \cite{schneeloch2014uncertainty}. This very interesting conjecture, which the authors called the `CQC conjecture,' although appeared over a decade ago in 2014, still remains open. The authors showed that if proven, this conjecture will have nice applications in strengthening Berta et al.'s EU relation \cite{berta2010uncertainty}, witnessing quantum entanglement \cite{horodecki2009quantum}, and also help in proving the security of quantum key distribution protocols \cite{pirandola2020advances}. They have also proved this conjecture for all pure states, states with one maximally mixed subsystem, and if one of the measurements is minimally disturbing, that is, if the measurement operator commutes with one of the reduced subsystems. In this work, we show that the CQC conjecture holds for more states other than the ones mentioned earlier, by identifying a sufficient condition based on Coles and Piani's work \cite{coles2014improved}. We also identify a further application of the CQC conjecture in generalizing the Maassen-Uffink EU relation \cite{maassen1988generalized}, whose potential applications may be of independent interest. Furthermore, we extend this conjecture to cover more mutually unbiased basis measurements and prime dimensions, which we call the ECQC conjecture. Then, we identify a sufficient condition for this ECQC conjecture. Moreover, with explicit calculations, we demonstrate the validity of this extended conjecture on isotropic states with arbitrary prime dimensions. We also perform experimental simulations of the validity of the ECQC conjecture on random pure and bipartite states in dimensions 3 and 5, where we do not observe any contradiction to this conjecture. 

	\section{Notation and Preliminaries}
In this work, we write $\rho_{AB}$ to indicate a bipartite state shared between two parties $A$ and $B$. We denote the entropy $H$ of system $A$ by $H(A)$ and similarly use $H(B)$ for system $B$. For the joint entropy of $A$ and $B$, we use the notation $H(AB)$. Two mutually unbiased bases we often use in this work are the bases $\Z$ and $\X$. These bases are high-dimensional generalizations of the standard computational and Hadamard basis. That is, here, $\Z = \{\ket{0}, \ket{1}, ... , \ket{d - 1}\}$ for the considered finite dimension $d < \infty$, and similarly, $\X = \{\mathcal{F}\ket{0}, \mathcal{F}\ket{1}, ... , \mathcal{F}\ket{d - 1}\}$, where $\mathcal{F}$ is the $d$-dimensional Fourier matrix:

\begin{align*}
	\mathcal{F} =  \frac{1}{\sqrt{d}} \sum_{x, y = 0}^{d - 1} e^{\frac{2 \pi i x y}{d}} \kb{x}{y}. 
\end{align*}

If both subsystems $A$ and $B$ are measured in the basis $\Z$, the resulting state is denoted by $\rho_{Z^{AB}}$. The joint entropy in this post measurement state is denoted by $H(Z^AZ^B)$. Similarly, for any other measurement in the basis $\{M_i\}_{i = 0}^{d + 1}$, the joint entropy is denoted by $H(M_i^AM_i^B)$. We denote the mutual information between two parties $A$ and $B$ by $I(A : B) = H(A) + H(B) - H(AB)$. Similarly, after one or both parties make a measurement on their respective subsystems, we use $I(M_i^A : B)$ or $I(M_i^A : M_i^B)$ to represent the decreased mutual information, respectively. For conditional entropies, we denote $H(A | B)$ to denote the conditional entropy of $A$ given access to the $B$ subsystem. After the measurement, the conditional entropies are represented by $H(M_i | B)$ or $H(M_i^A | M_i^B)$, depending on one- or two-sided measurements. To trace out a subsystem $A$ ($B$) from a bipartite state $\rho_{AB}$, we use the notation $\Tr_A(\rho_{AB}) = \rho_{B}$ ($\Tr_B(\rho_{AB}) = \rho_{A}$). Sometimes we denote a maximally entangled state by $\ket{\Psi^+}_{AB}$, without mentioning the subsystems $A$ and $B$'s dimensions, but it should be clear from the context. To denote the set of natural numbers from $1$ to $n$, we use the notation $[n]$, and to denote the subsets of a specific size $k$ of this set, we use $[n]^k$. All logarithms used in this work are of base $2$, and we also adopt the standard convention of information theory that $0 \log(0) \equiv 0$.

	\subsection{Mutually Unbiased Bases}
For an excellent review on MUBs, readers are encouraged to consult Durt et al.'s work \cite{durt2010mutually}. Informally, two measurement bases are said to be mutually unbiased if measuring a quantum state in one of these bases leaves us as maximally uncertain about the measurement outcomes in the other basis. More formally, we say that two measurement bases $\Z = \{ \ket{z_0}, \ket{z_1}, ... \ket{z_{d - 1}} \}$ and $\X = \{ \ket{x_0}, \ket{x_1}, ... \ket{x_{d - 1}} \}$ are mutually unbiased if it holds that:

\begin{align*} 
	|\bk{z_i}{x_j} |^2 = \frac{1}{d}, \text{ for any } i, j \in \{0, 1, ... d - 1	\}. 
\end{align*}

A set of such measurement bases are said to be mutually unbiased if they are pairwise mutually unbiased. This class of measurements is important in a wide range of information processing tasks, including quantum state tomography, where MUBs give the minimal and optimal set of measurements \cite{wootters1989optimal} to reconstruct the state's density operator. Although it is an interesting open question to determine the maximum number of MUBs in a dimension $d$, it is known that there are $d + 1$ MUBs if $d$ is a prime power \cite{durt2010mutually, wootters1989optimal, bandyopadhyay2002new}.

	\subsection{Entropic Uncertainty Relations}
Entropic uncertainty relations are quite integral in the study of quantum information theory in general and its various applications, such as analysis of quantum cryptographic protocols. See \cite{Coles2017Feb, Wehner2010Feb} for comprehensive surveys on this topic. As an example, we can see the well-known EU relation by Maassen and Uffink \cite{maassen1988generalized}, whose strengthened version by Coles et al. \cite{Coles2017Feb} we present below. The Maassen-Uffink relation states that for a quantum state $\rho_{A}$, if one performs two MUB measurements $\Z$ and $\X$ on it, the resulting uncertainty will have the following lower bound:

\begin{align}	
	H(Z^A) + H(X^A) \ge \log(d) + H(A). \label{eq:massuff}	
\end{align}

However, quantum systems could be correlated with other quantum systems too. In that case, conditioning on those correlated systems gives rise to interesting uncertainty relations as shown by Berta et al. \cite{Berta2010Sep}. They have shown that if the subsystem $A$ of a bipartite system $\rho_{AB}$ is arbitrarily correlated with some quantum system $B$, then performing $\Z$ and $\X$ basis measurements on $A$ will have the following lower bound for conditional entropies:

\begin{align}	
	H(Z^A | B) + H(X^A | B) \ge \log(d) + H(A | B), \label{eq:berta}	
\end{align}

where $d$ is the number of possible outcomes in the $\Z$ and $\X$ basis measurement, and also the dimension of the subsystems $A$ and $B$. Note that, in this case, we are considering $\Z$ and $\X$ to be MUB, but the original formulation by Berta et al. \cite{berta2010uncertainty} holds for any pair of observables, and a slight modification comes on the right-hand side of inequality \eqref{eq:berta}. Among the generalizations of the Maassen-Uffink relation to involve more than two MUB measurements, the one that we use in this work is found in one of the works by Sanchez \cite{sanchez1993entropic}. In this generalization, if we consider $d$ MUBs given that the dimension $d$ is a prime number, then the sum of the entropies of the resulting classical random variables $M_i^A$ for all $i \in \{1, ..., d + 1\}$ is as follows: 

\begin{align}
	\sum_{i = 1}^{d + 1} H(M_i^A) \ge (d + 1)(\log(d + 1) - 1) \label{eq:massuffd},	
\end{align}

where we see the usage of the fact that there are $d + 1$ MUBs if $d$ is a prime power as discussed in the previous section.

	%
	\section{The CQC conjecture} \label{sec:cqc}
	The CQC conjecture \cite{schneeloch2014uncertainty} asserts that the sum of two classical mutual information, obtained by bipartite measurements in the $\Z$ and $\X$ bases, is always upper bounded by the original mutual information. That is, on an arbitrary bipartite state $\rho_{AB}$, if we perform two measurements $\Z$ and $\X$ on both of the parties, we have that:
	\begin{align*}
		\izazb + \ixaxb \le \iab,
	\end{align*}
	where $I(.)$ is the mutual information between the random variables. It is noted in \cite{henderson2001classical} that, mutual information holds all types of classical and quantum correlations shared among the parties. In that sense, it is interesting to study the reasons for the validity of this conjecture. Because as noted in the work by Coles et al. \cite{Coles2017Feb}, the following does not always hold:
	\begin{align}
		I(Z^A : B) + I(X^A : B) \le \iab. \label{eq:doesnthold}
	\end{align}
	So the CQC conjecture, one could argue that, captures the reduction of mutual information when the conditioning subsystem makes the transition from the quantum to the classical realm. Along with this particular transition being an interesting study on its own, the original paper on CQC conjecture by Schneeloch et al. \cite{schneeloch2014uncertainty} also mentions important applications. For example, this conjecture can strengthen Berta's uncertainty relation \cite{berta2010uncertainty}, whose formulation we shall see later in this work. It also helps witnessing entanglement as the authors of the CQC conjecture showed. If CQC conjecture holds, then we have $I(A : B) \le \min \{H(A), H(B)\}$, and if it turns out that $\izazb + \ixaxb > \min \{H(A), H(B)\}$, then this implies that $H(A | B) < 0$, signifying entanglement between $A$ and $B$. The CQC conjecture also helps in proving security of quantum key distribution protocols as it places the following upper bound on the adversary Eve's ($E$) mutual information with one of the honest party Alice ($A$) \cite{devetak2005distillation}:
	\begin{align*}
		I(A : E) \le 2 \log(d) - \izazb - \ixaxb.
	\end{align*}
	Consequently, the quantum key distribution protocol designers can use this upper bound in establishing a rate for secure keys in Alice and Bob's communication \cite{pirandola2020advances}. Here, we identify one other possible application of the CQC conjecture. In the literature, we have seen generalizations of the Maassen-Uffink uncertainty relation, in the forms of involving multiple MUBs or expressions in terms of conditional entropies involving quantum memories \cite{Coles2017Feb, Berta2010Sep}. Interestingly, if the CQC conjecture holds, we get another generalization of the Maassen-Uffink relation to two parties in the following form, which has not been discussed previously as far as we know. 
	\begin{proposition}\label{prop:appcqc1}
		For a bipartite state $\rho_{AB}$, if we consider two mutually unbiased measurements $\Z$ and $\X$ in dimension $d = 2$, then, given the CQC conjecture, we have that
		\begin{align*}
			H(Z^AZ^B) + H(X^AX^B) \ge 2 + H(AB).
		\end{align*}
	\end{proposition}
	\begin{proof}
		We start from the CQC conjecture and see that proposition \eqref{prop:appcqc1} follows directly from it. 
		\begin{align}
			&\hspace{1cm} \izazb + \ixaxb \le \iab \nonumber\\
			&\implies \hza + \hzb - \hzazb + \hxa + \hxb - \hxaxb  \le \iab \nonumber\\
			&\implies \hza + \hxa + \hzb + \hxb- \hzazb - \hxaxb  \le \iab \nonumber\\		
			&\implies 1 + H(A) + 1 + H(B) - \hzazb - \hxaxb  \le H(A) + H(B) - H(AB) \nonumber\\
			&\implies H(Z^AZ^B) + H(X^AX^B) \ge 2 + H(AB), \label{eq:appcqc}
		\end{align}
		
		where in the penultimate line, we have used the Maassen-Uffink relation mentioned in inequality \eqref{eq:massuff}, strengthened by Coles et al. \cite{Coles2017Feb}. We note that the potential applications of this generalized Maassen-Uffink could be investigated in the future. \end{proof}
	
	\subsection{Extended CQC conjecture}
	The motivation to extend the CQC conjecture to cover more MUB measurements and prime dimensions comes from the discussion presented in the excellent review paper by Coles et al \cite{Coles2017Feb}. In that work, the authors list instances of works by various authors who showed that using mutually unbiased bases is not a guarantee to obtain strong EU relations. This weakness in the lower bound of EU relations ultimately results in providing further credence to the ECQC and especially to the CQC conjecture, as in both cases we are using less than the maximum number of MUBs available in the considered dimensions. Note that the CQC conjecture uses two MUBs, and even in dimension $2$, there are $3$ MUBs available \cite{durt2010mutually}. So considering more MUBs but that is always at least one less than the total number of MUBs available in a particular dimension $d$, we propose the following.
	\begin{conjecture} 
		Let $\rho_{AB}$ be a bipartite state where each subsystem has prime dimension $d$. If both parties make $d + 1$ measurements on their individual subsystems, then the mutual information of the resulting random variables will have the following upper bound: 
		\begin{align}
			I(A : B) \ge \underset{\mathcal{S} \in [d + 1]^d}{\min} \sum_{i \in \mathcal{S}} I(M^A_i : M^B_i). \label{conj:ecqc}
		\end{align}
	\end{conjecture}
	\noindent The elements of the set $\mathcal{S}$ are collections of $d$ mutual information resulting from $d + 1$ bipartite MUB measurements. So, the right-hand side of this inequality \eqref{conj:ecqc} captures the quantity where both parties make $d + 1$ number of MUB measurements and select the $d$ among those resulting classical mutual information that gives them the smallest sum. Inspired by the simulation on random states performed by the authors of the CQC conjecture, we performed similar simulations on random states and multiple dimensions and did not find any contradiction.
	\newline\newline
	\noindent As an immediate application of this extension, as we have seen in the case of CQC, we will get a generalization of the Maassen-Uffink uncertainty relation in the bipartite case for multiple MUB measurements and prime dimensions. We state that in the form of the following proposition. 	 
	\begin{proposition}
		For a bipartite state $\rho_{AB}$, where each subsystem $A$ and $B$ has prime dimension $d$, for a set of mutually unbiased bases $\{M_i\}_{i = 1}^{d + 1}$, given the ECQC conjecture, we have that:
		\begin{align}
			\sum_{i = 1}^{d + 1}\hmamb \ge 2((d + 1)(\log(d) - 1)) - \iab - \imax, \label{eq:appecqc}
		\end{align} 		 	
		where $I_{max}(M^A : M^B)  = \underset{i}{\max } \text{ } I(M^A_i : M^B_i), \forall \{i = 1, 2, ..., d + 1\}$, for any measurement base $M_i$, from the set of MUBs in this dimension $d$.	 	
	\end{proposition}
	\begin{proof}
		We follow the same strategy  as proposition \eqref{eq:appcqc} to arrive at proposition \eqref{eq:appecqc}. We start with a generalization of the Maassen-Uffink relation developed by Sanchez \cite{sanchez1993entropic}. This establishes the following lower bound on the sum of entropies resulting from MUB measurements on a system with prime dimension $d$, where $d + 1$ MUB measurements exist \cite{durt2010mutually}:
		\begin{align*}
			\sum_{i = 1}^{d + 1} H(M_i) \ge (d + 1)  (\log(d + 1) - 1). 
		\end{align*}
		Using this lower bound we see the following, given our conjecture:  
		\begin{align*}
			& \underset{\mathcal{S} \in [d + 1]^d}{\min} \sum_{i \in \mathcal{S}} I(M^A_i : M^B_i) \le I(A : B) \\ 
			&\implies \underset{\mathcal{S} \in [d + 1]^d}{\min} \sum_{i \in \mathcal{S}} I(M^A_i : M^B_i) + \imax \le I(A : B) + \imax\\ 	
			&\implies  \sum_{i = 1}^{d + 1} I(M^A_i : M^B_i) \le I(A : B) + \imax\\ 		
			&\implies \sum_{i = 1}^{d + 1} \hma + \sum_{i = 1}^{d + 1} \hmb - \sum_{i = 1}^{d + 1}\hmamb \le I(A : B) + \imax \\
			&\implies 2((d + 1)  (\log(d + 1) - 1)) - \sum_{i = 1}^{d + 1}\hmamb \le I(A : B) + \imax \\		
			&\implies \sum_{i = 1}^{d + 1}\hmamb \ge 2((d + 1)  (\log(d + 1) - 1)) - \iab - \imax. 	
		\end{align*}
		
	\end{proof}
	\section{Results}
	In this section, we present the results related to the CQC and ECQC conjectures. First we show that there is a sufficient condition for the validity of the CQC conjecture that covers states that are currently not known to hold CQC. This new class of states could be of arbitrary dimension $d$, not just the prime dimensions that are considered for the ECQC conjecture. 
	\subsection{A sufficient condition for CQC}
	\begin{proposition}\label{prop:suffcqc}
		For any bipartite state $\rho_{AB}$ and two MUB measurements $\Z$ and $\X$, if 
		\begin{align}
			\izab - \izazb + I(X^A:B) - I(X^A: X^B)  \ge \log(d) - H(A), \label{eq:suffcqc}
		\end{align}
		then the CQC conjecture holds. 
	\end{proposition}
	\begin{proof}
		In the article by Coles and Piani \cite{coles2014improved}, we see that for a bipartite state $\rho_{AB}$, if we apply two MUB measurement $\Z$ and $\X$ on subsystem $A$, then the sum of the two resulting mutual information has the following upper bound: 
		\begin{align}
			I(Z^A : B) + I(X^A : B)  \le \log(d) - H(A | B). \label{eq:cp}
		\end{align}
		Note that inequality \eqref{eq:cp} that we have presented here is a special case involving MUB measurements only, of the more general form of the equation that deals with arbitrary orthonormal bases from the article \cite{coles2014improved}. The general equation had the following form, written in our notations: 
		\begin{align*}
			I(Z^A : B) + I(X^A : B)  \le r - H(A | B), 
		\end{align*}
		where 
		\begin{align*}
			&r = min\{r(X, Z), r(Z, X)\}, \\
			&r(X, Z) = \log_2\left(d \sum_j \underset{k}{max} |\bk{x_j}{z_k}|^2\right), \\ 
			&r(Z, X) = \log_2\left(d \sum_k \underset{j}{max} |\bk{x_j}{z_k}|^2\right), 
		\end{align*}
		and the $x_j, z_k$ are the basis states of the $\X$ and $\Z$ basis measurements. Clearly, for MUBs, $r(Z, X) = r(X, Z) = \log(d)$ as $|\bk{x_j}{z_k}|^2 = \frac{1}{d}$ for MUBs, thus we get inequality \eqref{eq:cp}. From this inequality \eqref{eq:cp}, if we make a measurement on the B part, we get the following: 
		\begin{align}
			I(Z^A : Z^B)+ I(X^A : X^B)  \le \log(d) - H(A | B) \label{eq:cpm}, 
		\end{align}
		because of the data processing inequality \cite{wilde2011classical} which ensures $I(Z^A: Z^B) \le I(Z^A: B), I(X^A: X^B) \le I(X^A: B)$. On the other hand, the CQC conjecture \cite{schneeloch2014uncertainty} asks whether the following holds:
		\begin{align}
			I(Z^A : Z^B) + I(X^A : X^B) \le I(A : B). \label{eq:cqc}
		\end{align} 
		We also notice that $\log(d) - H(A | B) \ge I(A : B)$ always holds as $H(A) \le \log(d)$. These observations give us an idea about a possible sufficient condition. That is, going from inequality \eqref{eq:cp}, which is already known to be true, to inequality \eqref{eq:cqc}, which we are studying, we need to look at how much mutual information decreases on either sides of these two equations. If the reduction in mutual information on the left side of inequality \eqref{eq:cqc} is higher than the reduction on the right side, then the CQC conjecture would hold.  Let us start with the known inequality \eqref{eq:cp} and derive the CQC conjecture mentioned in inequality \eqref{eq:cqc}, assuming sufficient condition from inequality \eqref{eq:suffcqc}:
		\begin{align*}
			&I(Z^A : B) + I(X^A : B)  \le \log(d) - H(A | B) \\
			\implies &I(Z^A : B) + I(X^A : B) + \izazb - \izazb + \ixaxb - \ixaxb  \\
			&\le \log(d) - H(A | B) \text{ [adding and subtracting the same terms]}\\
			\implies & I(Z^A : B) - \izazb + I(X^A : B) - \ixaxb  + \izazb  + \ixaxb \\
			&\le \log(d) - H(A | B) \text{ [rearranging terms]}\\
			\implies &  \log(d) - H(A) + \izazb  + \ixaxb  \le \log(d) - H(A | B) \text{ [using sufficent condition \eqref{eq:suffcqc}]}\\
			\implies &  \izazb  + \ixaxb  \le \log(d) - H(A | B) - \log(d) + H(A) \\
			\implies &  \izazb  + \ixaxb  \le I(A:B).  \\
		\end{align*}
	\end{proof}
	
	\noindent As a demonstration of this sufficient condition, in figure \eqref{fig:suffcond2dim}, we consider a number of random density operators where each subsystem has dimension $d = 2$, which respect condition in \eqref{eq:suffcqc}, and are not pure, or entangled, or have a  maximally mixed subsystem, which are already known not to contradict CQC. So this sufficient condition adds a new class of states that are good for this conjecture. 
	\begin{figure}[h!]
		\centering
		\includegraphics[width=0.5\linewidth]{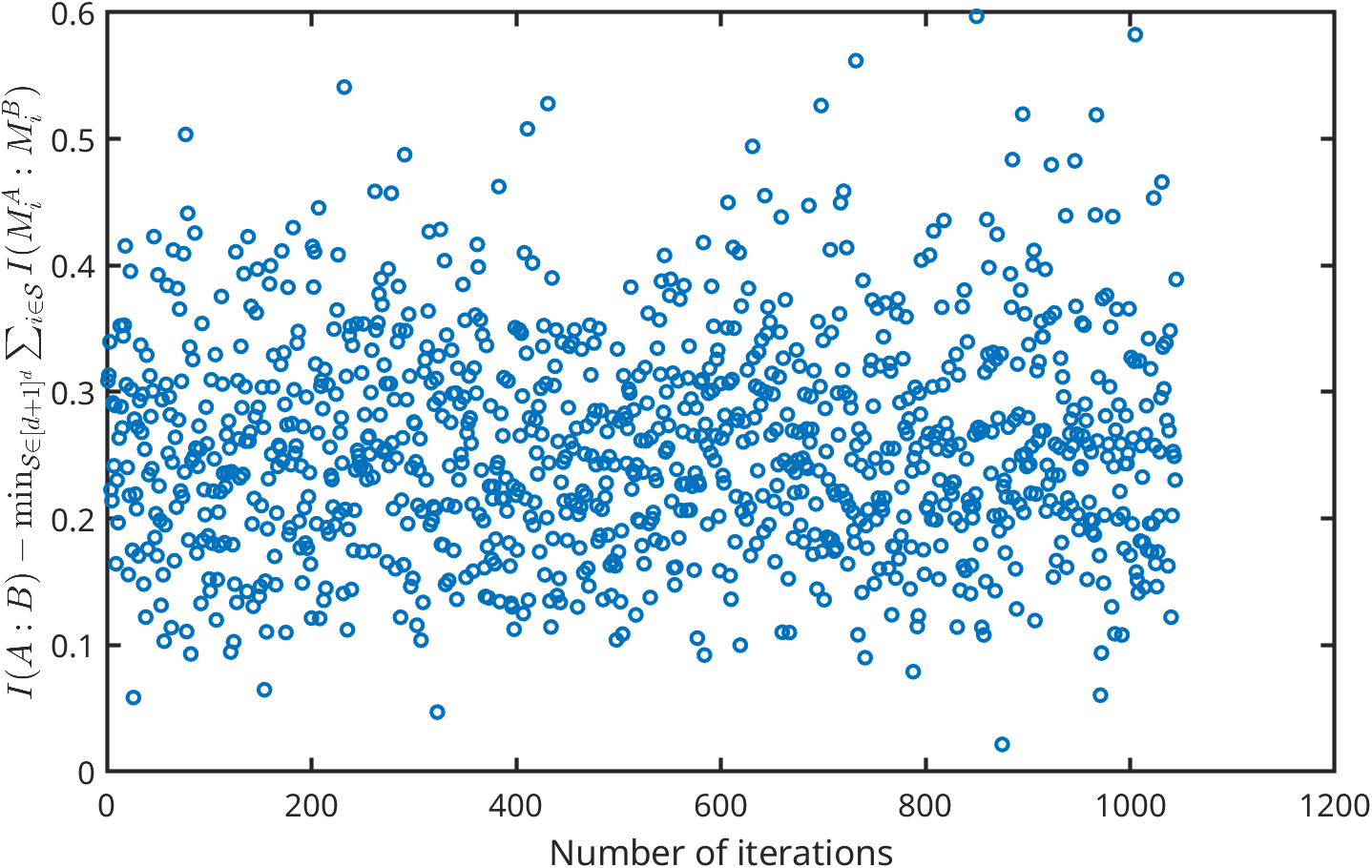}
		\caption{Simulating CQC conjecture with random separable mixed states that do not have a maximally mixed subsystem, but that respect inequality \eqref{eq:suffcqc}. This figure shows that the differences shown in the $Y$ axis are always non-negative.} 
		\label{fig:suffcond2dim}
	\end{figure}
	\noindent In this figure, we see that the $Y$ axis, which represents the difference between original mutual information before the measurement and the sum of mutual information resulting from two different MUB measurements, is never negative, signifying the validity of the CQC conjecture. 
	\newline\newline
	\noindent With this result in mind, we continue and find that there is a sufficient condition for the ECQC conjecture too, which ensures its validity for the cases when $A$ and $B$ both have the same prime dimension. We formally state this condition in the form of the following proposition. 
	\subsection{A sufficient condition for ECQC}
	\begin{proposition}\label{prop:suffecqc}
		For any bipartite state $\rho_{AB}$ where the subsystems $A$ and $B$ have prime dimension $d$, considering the set of MUBs  $\{M_i\}_{i = 1}^{d + 1}$ for dimension $d$, if 
		\small
		\begin{align}
			&\sum_{i = 1}^{d+1} I(M_i^A : B) - \sum_{i = 1}^{d + 1} I(M_i^A : M_i^B) 
			\ge \sum_{i = 1}^{d + 1} H(M_i) - \frac{d + 1}{2} \log(d) - \frac{d + 1}{2} H(A | B) - I(A:B), \label{eq:suffconecqc}
		\end{align}
		then ECQC holds. 		
	\end{proposition}
	
	\begin{proof}
		We follow the same strategy that gave us the sufficient condition for CQC, but for higher dimensional subsystems. That is, we start with the generalization of the Berta's uncertainty relation \cite{berta2010uncertainty} for prime dimensions which Xie  et al. obtained in \cite{xie2021optimized}. They showed the following lower bound on the conditional entropy given a quantum register, which we later use to arrive at a desired upper bound for the total mutual information when only one party measures in the MUBs. We state the inequality in a slightly different notation than used in the original work, which takes a special form for MUBs:
		\begin{align}
			&\sum_{i = 1}^{d + 1} H(M_i^A | B)  \ge \frac{d + 1}{2} \log d + \frac{d + 1}{2} H(A | B). \label{eq:xiefirst}
		\end{align}
		From this inequality \eqref{eq:xiefirst}, we have that:
		\begin{align}
			&\implies \sum_{i = 1}^{d + 1} H(M_i) - \sum_{i = 1}^{d + 1} I(M_i^A : B) \ge \frac{d + 1}{2} \log d + \frac{d + 1}{2} H(A | B) \nonumber\\
			&\implies \sum_{i = 1}^{d + 1} I(M_i^A : B) \le  \sum_{i = 1}^{d + 1} H(M_i) - \frac{d + 1}{2} \log d - \frac{d + 1}{2} H(A | B).   \label{eq:xie2}
		\end{align}
		Note that in the same article by Xie et al., they also present an optimized version of this inequality \eqref{eq:xiefirst}, where an additional term is added on the right hand side, and has the following form:
		\begin{align}
			&\sum_{i = 1}^{d + 1} H(M_i^A | B)  \ge \frac{d + 1}{2} \log d + \frac{d + 1}{2} H(A | B) + \max\{0, \delta_m\}. \label{eq:xieoptimized}
		\end{align}
		for $\delta_m = \frac{d + 1}{2}I(A : B) - \sum_{i = 1}^{d + 1} I(M_i^A : B)$. However, this optimized version is not applicable to our case as the two terms of $\delta_m$ has different multipliers. If the first term of $\delta_m$, without the multiplier $(d + 1)/2$ is already larger than the second term, then both CQC and ECQC hold trivially, and with the multiplier, this form is not compatible with our scenario as we have a single $\iab$ term, and we exclude it from our analysis. 
		\newline\newline
		Now, we mention our conjecture mentioned in inequality \eqref{conj:ecqc} for the clarity of derivation:
		\begin{align}
			I(A : B) \ge \underset{\mathcal{S} \in [d + 1]^d}{\min} \sum_{i \in \mathcal{S}} I(M^A_i : M^B_i).    \label{eq:conj2}
		\end{align}			
		Note that, because our conjecture considers $d$ measurements out of $d + 1$, and excludes the maximum mutual information obtained by one of the $d + 1$ MUB measurements, adding that quantity will give us the total $d + 1$ mutual information, and if one party is not measured, then the mutual information could only be higher because of the data processing inequality. So from the right hand side of inequality \eqref{eq:conj2} we have that:
		\begin{align}
			\underset{\mathcal{S} \in [d + 1]^d}{\min} \sum_{i \in \mathcal{S}} I(M^A_i : M^B_i) \le \sum_{i = 1}^{d + 1} I(M_i^A : M_i^B)  \le \sum_{i = 1}^{d + 1} I(M_i^A : B) \label{eq:conj2xie}. 
		\end{align}
		Now, we follow a similar strategy that we used for the sufficient condition for the original CQC conjecture. We start from inequality \eqref{eq:xie2} by Xie et al. and then arrive at the ECQC conjecture, assuming the sufficient condition from proposition \eqref{prop:suffecqc}. 
		\begin{align*}
			&\sum_{i = 1}^{d + 1} I(M_i^A : B) \le  \sum_{i = 1}^{d + 1} H(M_i) - \frac{d + 1}{2} \log d - \frac{d + 1}{2} H(A | B) \\
			\implies &\sum_{i = 1}^{d + 1} I(M_i^A : B) - \sum_{i = 1}^{d + 1} I(M_i^A : M_i^B) + \sum_{i = 1}^{d + 1} I(M_i^A : M_i^B) \\
			&\le  \sum_{i = 1}^{d + 1} H(M_i) - \frac{d + 1}{2} \log d - \frac{d + 1}{2} H(A | B) \text{ [adding and subtracting the same terms]}\\
			\implies &\sum_{i = 1}^{d + 1} H(M_i) - \frac{d + 1}{2} \log(d) - \frac{d + 1}{2} H(A | B) - I(A:B) + \sum_{i = 1}^{d + 1} I(M_i^A : M_i^B) \\
&\le  \sum_{i = 1}^{d + 1} H(M_i) - \frac{d + 1}{2} \log d - \frac{d + 1}{2} H(A | B) \text{  [using sufficient condition \eqref{eq:suffconecqc}]}\\
\implies & I(A:B) \ge \sum_{i = 1}^{d + 1} I(M_i^A : M_i^B) \\
\implies &I(A:B) \ge \underset{\mathcal{S} \in [d + 1]^d}{\min} \sum_{i \in \mathcal{S}} I(M^A_i : M^B_i) \text{[ from inequality \eqref{eq:conj2xie}]}. 
		\end{align*}
	\end{proof}
	

	\subsection{Attempt to establish ECQC directly for isotropic states}	
	In the original paper, Schneeloch et al. \cite{schneeloch2014uncertainty} showed that the CQC conjecture holds for pure states, states with one maximally  mixed subsystem and if one of the measurements is minimally disturbing for one of the parties. The validity of the CQC conjecture for the latter two types of states was shown by relating them to Berta et al.'s EU relation mentioned in inequality \eqref{eq:berta}. 
	With some algebra, the authors reformulate the CQC conjecture to arrive at:
	\begin{align}
		H(Z^A | Z^B) + H(X^A | X^B) &\ge \log(d) + H(A | B) + \underbrace{[H(Z^A) + H(X^A) - \log(d) - H(A)]}_{\kappa_1}. \label{eq:cqcreform}
	\end{align} 
	In inequality \eqref{eq:cqcreform}, we denote the terms inside the parentheses `[]' by $\kappa_1$. The authors point out that if $\kappa_1 = 0$, then the CQC conjecture holds as we can measure the second subsystem in Berta et al.'s EU relation involving quantum memories to arrive at \eqref{eq:cqcreform}. Here, we recall Maassen-Uffink uncertainty relation \cite{maassen1988generalized, Coles2017Feb}:
	\begin{align*}
		H(Z^A) + H(X^A) \ge \log(d) + H(A).
	\end{align*}
	It is easy to see that states makes this Maassen-Uffink EU relation an equality, namely, states which are equal to the eigenvectors of those two measurement bases, also satisfies the CQC conjecture as that would make $\kappa_1 =  0$. 
	%
	In the following, we show that this strategy is not applicable to ECQC. That is, even though we can arrive at a similar inequality as \eqref{eq:cqcreform}, ECQC's validity cannot be established for similar states. We start by assuming the validity of the ECQC conjecture:
	\begin{align}
		&\underset{\mathcal{S} \in [d + 1]^d}{\min} \sum_{i \in \mathcal{S}} I(M^A_i : M^B_i) \le 		I(A : B) \nonumber\\
		&\implies \sum_{i = 1}^{d + 1} I(M^A_i : M^B_i) = I(A : B) + I_{max}(M^A : M^B) \nonumber\\
		&\implies \sum_{i = 1}^{d + 1} H(M_i^A) - \sum_{i = 1}^{d + 1} H(M_i^A | M_i^B) \le I(A : B) + I_{max}(M^A : M^B) \nonumber\\
		&\implies \sum_{i = 1}^{d + 1} H(M_i^A | M_i^B) \ge \sum_{i = 1}^{d + 1} H(M_i^A) - I(A : B) - I_{max}(M^A : M^B). \label{eq:ecqcberta}
	\end{align}
	In the second line of the derivation, we add $I_{max}(M^A : M^B)$ to both sides, which in our notation is the maximum mutual information obtained by any MUB measurement basis. Now we restate Xie et al.'s result \cite{xie2021optimized} that generalizes Berta et al.'s EU relation involving conditional entropy with quantum memories \eqref{eq:berta}:
	\begin{align}
		\sum_{i = 1}^{d + 1} H(M_i^A | B)  &\ge \frac{d + 1}{2} \log d + \frac{d + 1}{2} H(A | B). \label{eq:xie}
	\end{align}
	We add and subtract the right hand side of inequality \eqref{eq:xie} in the right hand side of inequality \eqref{eq:ecqcberta} to arrive at:
	\begin{align}
		&\sum_{i = 1}^{d + 1} H(M_i^A | M_i^B) \ge \frac{d + 1}{2} \log d + \frac{d + 1}{2} H(A | B) \nonumber\\ 
		& \underbrace{\left[ - \frac{d + 1}{2} \log d - \frac{d + 1}{2} H(A | B) + \sum_{i = 1}^{d + 1} H(M_i^A) - I(A : B) - I_{max}(M^A : M^B).\right]}_{\kappa_2} \label{eq:ecqcparen}
	\end{align}	
	We comment that the structural similarity between equations \eqref{eq:cqcreform} and  \eqref{eq:ecqcparen} is noticeable. In inequality \eqref{eq:ecqcparen}, we denote the terms inside the parentheses `[]' by $\kappa_2$. If $\kappa_2 = 0$, that is:
	\begin{align}
		\frac{d + 1}{2} \log d + \frac{d + 1}{2} H(A | B) + I(A : B) + I_{max}(M^A : M^B) = \sum_{i = 1}^{d + 1} H(M_i^A), \label{eq:exequality} 
	\end{align}
	then ECQC holds based on Xie et al.'s result \eqref{eq:xie}. It is easy to see that that states of the form $\rho_{AB} = \rho_{A} \otimes \rho_{B}$ where $\rho_{A} = \frac{\mbi}{d}$ satisfies inequality \eqref{eq:exequality} because for these product states, the left and right hand sides become $\frac{d + 1}{2}\log_2(d) + \frac{d + 1}{2} \log_2(d) + 0 + 0 = (d + 1) \log_2(d)$. Note that the maximally entangled state, which of course has maximally mixed subsystems, does not respect inequality \eqref{eq:exequality}, as in that case we get, $H(A | B) = -\log(d), I(A : B) = 2 \log(d), I_{max}(M^A : M^B) = H(M^A) + H(M^B) - H(M^AM^B) = \log(d)$, resulting in  $3 \log(d) \ne (d+1) \log(d)$, for prime $d \ge 3$. So we cannot decide whether ECQC holds  for maximally entangled states based on inequality \eqref{eq:ecqcparen} being true, unlike the CQC conjecture, where the equality condition of a similar equation sufficed. With this method's insufficiency, let us calculate the validity of ECQC for Isotropic states directly, which are maximally entangled for a certain parameter range. 
	
	\subsection*{Isotropic states}
	In this section, we show the validity of ECQC for isotropic states with prime dimensions by an explicit calculation. Isotropic states are defined as \cite{springer1972springer, terhal2000entanglement}:
	\begin{align}
		\rho_{AB} = p \times \kb{\Psi^+}{\Psi^+}_{AB} + (1 - p) \times \frac{\mathbb{I}_{AB}}{d^2}, \label{eq:isotropic}
	\end{align}
	where $\ket{\Psi^+}_{AB} = \frac{1}{\sqrt{d}} \sum_{i = 0}^{d - 1} \ket{ii}_{AB}$, $\frac{-1}{d^2 - 1} \le p \le 1$, $d \ge 3$ is a prime. For any $p < 1$ in the definition of isotropic states in inequality \eqref{eq:isotropic}, we see that for inequality \eqref{eq:exequality}, the $l.h.s < r.h.s$, which results in $\kappa_2 > 0$. Consequently, we cannot recover Xie et al.'s \cite{xie2021optimized} result mentioned in inequality \eqref{eq:xie} conclusively in inequality \eqref{eq:ecqcparen}. If it is possible to tighten Xie et al.'s result further, in that case this could be a feasible way to show the validity of ECQC conjecture in this case. So we must perform a direct calculation for isotropic states to see if the ECQC conjecture holds. 
	
	In order to calculate the mutual information in this state, we have to identify the eigenvalues of $\rho_{AB}$ as well as of the subsystems $A$ and $B$. Below, we denote the vector of eigenvalues of $\rho_{AB}$ by $\text{eig}(\rho_{AB})$. 
	\begin{align*}
		\text{eig}(\rho_{AB}) = p \times \text{eig}\left(\kb{\Psi^+}{\Psi^+}_{AB}\right) + (1 - p) \times \text{eig}\left(\frac{\mathbb{I}_{AB}}{d^2} \right). 
	\end{align*}
	We know the eigenvalues of the maximally entangled state $\kb{\Psi^+}{\Psi^+}_{AB}$ and the completely mixed state $\frac{\mathbb{I}_{AB}}{d^2}$. From this, we easily see that the $d^2$ number of eigenvalues of $\rho_{AB}$ are $\text{eig}(\rho_{AB}) =\{p + \frac{1 - p}{d^2}, \frac{1 - p}{d^2}, \frac{1 - p}{d^2}, ..., \frac{1 - p}{d^2}\}$. So we get the following joint entropy of $\rho_{AB}$ where both subsystems $A$ and $B$ have dimension $d$:
	\begin{align*}
		H(\rho_{AB}) = - \left(p + \frac{1 - p}{d^2}\right) \log \left(p + \frac{1 - p}{d^2}\right) - \sum_{i = 1}^{d^2 - 1} \frac{1 - p}{d^2} \log \frac{1 - p}{d^2}.
	\end{align*}
	The subsystems of isotropic states $A$ and $B$ also has known eigenvalues, as they are simply the maximally mixed states with eigenvalues $\frac{1}{d}$ as long as $\frac{1}{d + 1} < p \le 1$, in which range, isotropic states are entangled \cite{krammer2005quantum, yang2021decompositions}. Even in the range of $p$ where the state is separable, (that is $-\frac{1}{d^2 - 1} \le p \le \frac{1}{d + 1}$), the partial traces $A$ and $B$ are still maximally mixed, which can be seen from direct calculation. 
	\begin{align*}
		\Tr_A(\rho_{AB}) &= \Tr_A\left(p \times \kb{\Psi^+}{\Psi^+}_{AB} + (1 - p) \times \frac{\mathbb{I}_{AB}}{d^2}  \right) \\
		&= \left(\frac{p}{d} + \frac{1 - p}{d}  \right) \mbi_B  = \frac{\mbi_B}{d},
	\end{align*}
	and we can see that $\Tr_B(\rho_{AB}) = \frac{\mbi_A}{d} $ in a similar way. So in these states,
	\begin{align*}
		\text{eig}(A) = \text{eig}(B) = \left(\frac{1}{d}, \frac{1}{d},..., \frac{1}{d}\right),
	\end{align*}
	for all $-\frac{1}{d^2 - 1} \le p \le 1$. From these calculations, we see that $H(A) = H(B) = \log(d)$, and the mutual information for isotropic states for any prime dimensions is: 
	\begin{align}
		I(A : B) &= H(A) + H(B) - H(\rho_{AB}) \nonumber\\
		&=2  \log(d) - \left(- \left(p + \frac{1 - p}{d^2}\right) \log \left(p + \frac{1 - p}{d^2}\right) - \sum_{i = 1}^{d^2 - 1} \frac{1 - p}{d^2} \log \frac{1 - p}{d^2}\right) \nonumber\\
		&= 2  \log(d) + \frac{p(d^2 - 1) + 1}{d^2} \log \left( \frac{p(d^2 - 1) + 1}{d^2} \right) + \sum_{i = 1}^{d^2 - 1} \frac{1-p}{d^2} \log \frac{1-p}{d^2} \label{eq:genmutinf}. 
	\end{align}
	Now we calculate the mutual information in the computational $\Z$ basis measurement, which is part of the MUBs. Because isotropic states have maximally mixed subsystems, the entropy of these subsystems after the measurement is still maximum. That is, $H(M_i^A) = H(M_i^B)$ for any MUB measurement $M_i$. However, the mutual information of the post-measurement state does change, and we calculate that below. 
	\subsubsection{Mutual information in isotropic states for dimension 3}
	For dimension 3, we consider the four MUBs given in \cite{brierley2009mutually}. These matrices, whose columns represent the basis states of the respective basis, are the following:
	\small
	\begin{align}
		M_0 = 
		\begin{bmatrix}
			1 & 0 & 0 \\
			0 & 1 & 0 \\
			0 & 0 & 1
		\end{bmatrix}, 
		M_1 = \frac{1}{\sqrt{3}} 
		\begin{bmatrix}
			1  & 1 & 1 \\
			1 & \om & \om^2 \\
			1 & \om^2 & \om
		\end{bmatrix}, 
		M_2 = \frac{1}{\sqrt{3}} 
		\begin{bmatrix}
			1 & 1 & 1 \\
			\om^2 & 1 & \om \\
			\om^2 & \om & 1
		\end{bmatrix},
		M_3 = \frac{1}{\sqrt{3}} 
		\begin{bmatrix}
			1 & 1 & 1 \\
			\om & \om^2 & 1 \\
			\om & 1 & \om^2	
		\end{bmatrix} \label{eq:dim3mats}	
	\end{align}
	Clearly, these matrices follow the restriction of being mutually unbiased to each other, as taking the absolute value squared of the inner product of any pair of columns from two different matrices here, produces $\frac{1}{3}$. That is, if $\ket{x} \in M_i$ and $\ket{y} \in M_j$ are any two columns for $i \ne  j$ in inequality \eqref{eq:dim3mats}, we have $|\bk{x}{y}|^2 = \frac{1}{3}$. 
	\paragraph{Calculation of $\izazb$:}
	First let us calculate the mutual information of $\rho_{Z^{AB}}$, which is the state obtained after $A$ and $B$ make measurements on their individual subsystems in basis $\Z$. As per the definition of mutual information: 
	\begin{align*}
		I(Z^A : Z^B) = H(Z^A) + H(Z^B) - H(Z^A Z^B), 
	\end{align*}
	where, $Z^A$ and $Z^B$ are the resulting random variables of their individual measurements. These random variables of the subsystems has maximum entropy as discussed in the previous section, that is (for $d = 3$): 
	\begin{align*}
		H(Z^A) = H(Z^B) = \log(3).
	\end{align*}
	So to get $\izazb$, we calculate the joint entropy term $H(Z^A Z^B)$, which is:
	\begin{align}
		H(Z^A Z^B) = - \sum_{i = 0, j = 0}^{2} p_{ij} \log (p_{ij}), \label{eq:jointzps}
	\end{align}
	where each $p_{ij}$ represent the probabilities of Alice and Bob's measurement outcomes to be any of the basis states in the $\Z$ basis measurement. Our task then reduces to calculating these $p_{ij}$ values in inequality \eqref{eq:jointzps}. We can directly calculate them as we show below. We see that $A$ and $B$ could only have outcomes $\ket{0}, \ket{1}, \ket{2}$ in their $\Z$ basis measurement as we are considering $d = 3$. We start with the calculation of the probability, that their individual measurements in this basis on isotropic state \eqref{eq:isotropic}, both produced $\ket{0}$, denoted by $p_{00}$, considering $\ket{\Psi^+}_{AB} = \frac{1}{\sqrt{3}} (\ket{00} + \ket{11} + \ket{22})$. 
	\begin{align*}		
		p_{00} &= p \bra{00} \kb{\Psi^+}{\Psi^+}_{AB} \ket{00} + (1 - p) \bra{00} \frac{\mathbb{I}_{AB}}{9} \ket{00}\\
		&= p \bra{00} \frac{1}{3} (\kb{00}{00} + \kb{00}{11} + \kb{00}{22} + \kb{11}{00} + \kb{11}{11} + \kb{11}{22} + ...\\
		&\kb{22}{00} + \kb{22}{11} + \kb{22}{22})\ket{00} + (1 - p) \bra{00} \frac{\mathbb{I}_{AB}}{9} \ket{00}\\
		&= \frac{p}{3} \times (\bk{00}{00} \bk{00}{00}+ \bk{00}{00} \bk{11}{00} + \bk{00}{00} \bk{22}{00} + ... \\
		&\bk{00}{11}\bk{00}{00} + \bk{00}{11}\bk{11}{00} + \bk{00}{11}\bk{22}{00} + ...\\
		&\bk{00}{22}\bk{00}{00} + \bk{00}{22}\bk{11}{00} + \bk{00}{22}\bk{22}{00}) + (1 - p) \bra{00} \frac{\mathbb{I}_{AB}}{9} \ket{00} \\
		&= \frac{p}{3} + \frac{1 - p}{9}. 
	\end{align*} 
	Similarly, we can also calculate the cases $p_{11}$ and $p_{22}$ which would be the same as $p_{00}$. That is, $$p_{00} = p_{11}  = p_{22} = \frac{p}{3} + \frac{1 - p}{9}.$$Now we focus on the probabilities of their mismatched measurement outcomes in this basis, denoted by $p_{ij}$ for $i \ne j$. For example, if we want $p_{01}$: 
	\begin{align*}		
		p_{01} &= p \bra{01} \kb{\Psi^+}{\Psi^+}_{AB} \ket{01} + (1 - p) \bra{01} \frac{\mathbb{I}_{AB}}{9} \ket{01}\\
		&=p \bra{01} \frac{1}{3} (\kb{00}{00} + \kb{00}{11} + \kb{00}{22} + \kb{11}{00} + \kb{11}{11} + \kb{11}{22} + ...\\
		&\kb{22}{00} + \kb{22}{11} + \kb{22}{22})\ket{01} + (1 - p) \bra{01} \frac{\mathbb{I}_{AB}}{9} \ket{01}. 
	\end{align*}
	Note that, for terms of the form $\kb{ii}{jj}$, $\bra{01}(\kb{ii}{jj}) \ket{01} = \bk{01}{ii} \bk{jj}{01} = 0$, as both of these inner products are zero for any $i, j \in \{0, 1, 2\}$. So we get:
	\begin{align*}
		p_{01} = (1 - p) \bra{01} \frac{\mathbb{I}_{AB}}{9} \ket{01} = \frac{1 - p}{9}.		
	\end{align*}
	With the same reasoning, we see that for any such pair $i \ne j$, $p_{ij} = \frac{1 - p}{9}$.
	Aggregating the probabilities for all $i, j \in \{0, 1, 2\}$, we see that for isotropic states in dimension $3$, the $\Z$ basis measurement outcome probabilities are: 
	\begin{align*}
		p_{ii} = \frac{p}{3} + \frac{1-p}{9}, \text{ } p_{i \ne j} = \frac{1-p}{9}. 
	\end{align*}
	From these probabilities, we calculate the mutual information of the post-measurement state after $A$ and $B$ make the $\Z$ basis measurement, namely $\rho_{Z^{AB}}$, denoted by $I\left(Z^A : Z^B\right)$. 
	\begin{align} 
		I\left(Z^A : Z^B\right) &= 2 \log(3) -  \left(-\sum_{i = 0}^{2} p_{ii} \log \left(p_{ii}\right) - \sum_{(i \ne j) =0 }^{2} p_{ij} \log \left(p_{ij}\right)\right) \nonumber\\
		&= 2 \log(3) -  \left(-3 \times \left(\frac{p}{3} + \frac{1-p}{9}\right)  \log \left(\frac{p}{3} + \frac{1-p}{9}\right) - 6 \times \frac{1-p}{9} \log \left(\frac{1-p}{9}\right) \right) \nonumber\\
		&= 2 \log(3) +  3 \times \left(\frac{2 p + 1}{9} \right)  \log \left(\frac{2p + 1 }{9} \right) + 6 \times \frac{1 - p}{9} \log \left(\frac{1-p }{9}\right). \label{eq:izazb}
	\end{align} 
	This concludes the calculation of $I\left(Z^A : Z^B\right)$. 
	\paragraph{Calculation of $I\left(X^A : X^B\right)$:}
	\noindent Now let us focus on the mutual information of the post measurement state $\rho_{X^{AB}}$, resulted from their $\X$ basis measurement. This quantity is denoted by $I\left(X^A : X^B\right)$. The calculation in this case is quite similar to the previous calculation for basis $\Z$. Here the possible measurement outcomes are, in other words, the basis vectors of the measurement $\X$ are: 
	\begin{align*}
		\ket{f_0} \equiv \mathcal{F} \ket{0} = \frac{1}{\sqrt{3}}(1, 1, 1), \ket{f_1} \equiv \mathcal{F} \ket{1} = \frac{1}{\sqrt{3}}(1, \omega,
		\omega^2), \ket{f_2}  \equiv \mathcal{F} \ket{2} = \frac{1}{\sqrt{3}}(1, \omega^2, \omega),	
	\end{align*}	
	where $\mathcal{F}$ is the Fourier matrix and the root of unity $\omega = e^{\frac{2 \pi i}{3}}$. First we calculate the probability that both $A$ and $B$ obtains outcome $\ket{f_0}$, denoted by $p_{f_0f_0}$. We see that
	\begin{align}
		p_{f_0f_0} &= p \bra{f_0f_0} \kb{\Psi^+}{\Psi^+}_{AB} \ket{f_0f_0} + (1 - p) \bra{f_0f_0} \frac{\mathbb{I}_{AB}}{9} \ket{f_0f_0} \nonumber\\
		&=p \bra{f_0f_0} \frac{1}{3} (\kb{00}{00} + \kb{00}{11} + \kb{00}{22} + \kb{11}{00} + \kb{11}{11} + \kb{11}{22} + ...\nonumber\\
		&\kb{22}{00} + \kb{22}{11} + \kb{22}{22})\ket{f_0f_0} + (1 - p) \bra{f_0f_0} \frac{\mathbb{I}_{AB}}{9} \ket{f_0f_0}. \label{eq:calcpf00part}
	\end{align}
	The full calculation is not too enlightening to write down, so we focus on a particular part of it below. Note that $\ket{f_0} = \frac{1}{\sqrt{3}}(\ket{0} + \ket{1} + \ket{2})$, so for any $i, j \in \{0, 1, 2\}$, the inner product between $\ket{i}, \ket{j}$ and $\ket{f_0}$ will simply be $\frac{1}{\sqrt{3}}$:
	\begin{align*}
		&\bra{f_0f_0}\kb{ii}{jj}\ket{f_0f_0} = \bk{f_0}{i}\bk{f_0}{i}\bk{j}{f_0}\bk{j}{f_0} = \left( \frac{1}{\sqrt{3}}\right)^4  = \frac{1}{9}.
	\end{align*}
	Summing up these terms in the calculation of $p_{f_0f_0}$ in inequality \eqref{eq:calcpf00part}, we get:
	\begin{align*}
		p_{f_0f_0} = \frac{p}{3} \left(9 \times\frac{1}{9}\right) + \frac{1 - p}{9} = \frac{p}{3} + \frac{1 - p}{9}.	
	\end{align*}
	Let us look at one other term where $(i + j) \% 3 = 0$, namely, $p_{f_1f_2}$. For this probability term, $\bk{f_1f_2}{00}\bk{00}{f_1f_2} = |\bk{f_1f_2}{00}|^2 = \frac{1}{9}$, $|\bk{f_1f_2}{11}|^2 = |\bk{f_1f_2}{22}|^2 = |\frac{\omega^3}{3}|^2 = \frac{1}{9}$.  For the terms $\bk{f_1f_2}{ii}\bk{jj}{f_1f_2}$,  where $i \ne j$, also gives us $\frac{1}{9}$ as we can easily see. For example, $\bk{f_1f_2}{00}\bk{11}{f_1f_2} = \bk{f_1}{0}\bk{f_2}{0}\bk{1}{f_1}\bk{1}{f_2} = (\frac{1}{\sqrt{3}})^4 \times 1 \times 1 \times \omega \times \omega^2 =  (\frac{1}{\sqrt{3}})^4 \times \omega^3 = \frac{1}{9}$ and $\bk{f_1f_2}{22}\bk{11}{f_1f_2} = \bk{f_1}{2}\bk{f_2}{2}\bk{1}{f_1}\bk{1}{f_2} = (\frac{1}{\sqrt{3}})^4 \times \overline{\om^2} \times \overline{\om} \times \om \times \om^2 = (\frac{1}{\sqrt{3}})^4 \times \omega \times \omega^2 \times \omega \times \omega^2 =  (\frac{1}{\sqrt{3}})^4 \times \omega^3 \times \omega^3 = \frac{1}{9}$, noting that, $\om^2 = \overline{\om}$ and $\om = \overline{\om^2}$, the complex conjugates and of course, $\om^3 = 1$. We can perform a similar calculation for $p_{f_2f_1}$ and arrive at the following conclusion:
	\begin{align*}
		p_{f_0f_0} = p_{f_1f_2} = p_{f_2f_1} = \frac{p}{3}\left(3 \times \frac{1}{9} + 6 \times \frac{1}{9}\right) + \frac{1 - p}{9} \left(9 \times \frac{1}{9}\right) = \frac{p}{3}  + \frac{1-p}{9}. 
	\end{align*}
	After considering the cases where $(i + j) \% 3 = 0$, let us see the cases where this is not true, that is, the terms where $(i + j) \% 3 \ne 0$. For example, the probability that $A$'s measurement outcome is $\ket{f_0}$ and $B$'s measurement outcome is $\ket{f_1}$, when they measure the isotropic state in this dimension:
	\begin{align}
		p_{f_0f_1} &= p \bra{f_0f_1} \kb{\Psi^+}{\Psi^+}_{AB} \ket{f_0f_1} + (1 - p) \bra{f_0f_1} \frac{\mathbb{I}_{AB}}{9} \ket{f_0f_1} \nonumber\\
		&=p \bra{f_0f_1} \frac{1}{3} (\kb{00}{00} + \kb{00}{11} + \kb{00}{22} + \kb{11}{00} + \kb{11}{11} + \kb{11}{22} + ...\nonumber\\
		&\kb{22}{00} + \kb{22}{11} + \kb{22}{22})\ket{f_0f_1} + (1 - p) \bra{f_0f_1} \frac{\mathbb{I}_{AB}}{9} \ket{f_0f_1}. \label{eq:calcpf00part2}
	\end{align}
	Note that the three terms of the form $\bra{f_0f_1}\kb{ii}{ii}\ket{f_0f_1} = |\bk{f_0f_1}{ii}|^2$ for $i \in \{0, 1, 2\}$ produces $|\frac{1}{3}|^2$, $|\frac{\omega}{3}|^2$, $|\frac{\omega^2}{3}|^2$ respectively, which are all equal to $\frac{1}{9}$. The other cases however, when we consider terms like $\bra{f_0f_1}\kb{ii}{jj}\ket{f_0f_1}$ for $i \ne j$, we get $\bra{f_0f_1}\kb{ii}{jj}\ket{f_0f_1} + 
	\bra{f_0f_1}\kb{jj}{ii}\ket{f_0f_1}= -\frac{1}{9}$ for $i, j \in \{0, 1, 2\}$, which can be seen by direct calculation. For example, 
	\begin{align*}
		\bra{f_0f_1}\kb{00}{11}\ket{f_0f_1} +  \bra{f_0f_1}\kb{11}{00}\ket{f_0f_1} = \frac{\om + \om^2}{9}= -\frac{1}{9},
	\end{align*} 
	from the sum of roots of unity property. So for the six terms where $i \ne j$, we get three $-\frac{1}{9}$. These negative quantities cancel out the three quantities where $i = j$. This results in the probability $p_{f_if_j}$ of outcomes where $(i + j) \% 3 \ne 0$ being a constant and equal to $\frac{1 - p}{9}$. Aggregating these probabilities, we see that:
	\begin{align*}
		p_{f_if_j} = \frac{p}{3} + \frac{1 - p}{9}, \text{ if }(i + j) \% 3 = 0, \text{ and } p_{f_if_j} = \frac{1 - p}{9} \text{ otherwise,}
	\end{align*}
	for $i, j \in \{0, 1, 2\}$. With the similar reasoning for $Z^A$ and $Z^B$, we see that $X^A$ and $X^B$ which are the resulting random variables after $A$ and $B$'s individual measurements, are still maximally mixed and we have the equal mutual information in this basis as the basis $\Z$:
	\begin{align} 
		I\left(X^A : X^B\right) &= 2 \log(3) +  3 \times \left(\frac{2 p + 1}{9} \right)  \log \left(\frac{2p + 1 }{9} \right) + 6 \times \frac{1 - p}{9} \log \left(\frac{1-p }{9}\right). \label{eq:ixaxb}
	\end{align} 
	\paragraph{Calculating $I(M_2^A : M_2^B)$:}
	We continue our calculation for the other two MUBs, namely, $M_2, M_3$ defined in \eqref{eq:dim3mats}, in this dimension $d = 3$ and see that they do not produce any mutual information. Let's calculate the probabilities of measurement outcomes for $M_2$ as described in the inequality \eqref{eq:dim3mats}. Here the basis vectors are: 
	\begin{align*}
		\ket{g_0} \equiv \obsth(1, \omega^2, \omega^2), \ket{g_1} \equiv \obsth(1 , 1, \omega), \ket{g_2} \equiv  \obsth (1, \omega, 1),
	\end{align*}
	for $\omega = e^{\frac{2 \pi i}{3}}$. Now, let us consider the probability that both $A$ and $B$ has the measurement outcome $\ket{g_0}$, denoted by $p_{g_0g_0}$. 
	\begin{align}
		p_{g_0g_0} &= p \bra{g_0g_0} \kb{\Psi^+}{\Psi^+}_{AB} \ket{g_0g_0} + (1 - p) \bra{g_0g_0} \frac{\mathbb{I}_{AB}}{9} \ket{g_0g_0} \nonumber\\		
		&=\frac{p}{3} \bra{g_0g_0} (\kb{00}{00} + \kb{00}{11} + \kb{00}{22} + \kb{11}{00} + \kb{11}{11} + \kb{11}{22} + ... \nonumber\\
		&\kb{22}{00} + \kb{22}{11} + \kb{22}{22})\ket{g_0g_0} + (1 - p) \bra{g_0g_0} \frac{\mathbb{I}_{AB}}{9} \ket{g_0g_0} \label{eqn:pg00}
	\end{align}		
	Focusing on one particular term $\bra{g_0g_0}\kb{00}{00}\ket{g_0g_0}$, we see that:
	\begin{align*}
		&\bra{g_0g_0}\kb{00}{00}\ket{g_0g_0} = \left(\obsth(1, \omega^2, \omega^2)\times (1, 0, 0)^T\right)^4 = \frac{1}{9}, 
	\end{align*}
	where `T' denotes transpose of that vector. In the case of $\bra{g_0g_0}\kb{11}{11}\ket{g_0g_0}$:
	\begin{align*}
		&\bra{g_0g_0}\kb{11}{11}\ket{g_0g_0} = \left|\obsth(1, \omega^2, \omega^2)\times (0, 1, 0)^T\right|^4 = \left|\frac{\omega^2}{\sqrt{3}}\right|^4  = \frac{1}{9}.
	\end{align*}
	Similarly we can see that, $\bra{g_0g_0}\kb{22}{22}\ket{g_0g_0} =\frac{1}{9}$. These are the cases when the inner products are of the form $\bra{g_0g_0}\kb{ii}{ii}\ket{g_0g_0}$, for $i \in \{0, 1, 2\}$. In the cases when $i \ne j$, in the inequality \eqref{eqn:pg00}, we have both $\bra{g_0g_0}\kb{ii}{jj}\ket{g_0g_0}$ and $\bra{g_0g_0}\kb{jj}{ii}\ket{g_0g_0}$ for $i, j \in \{0, 1, 2\}$. That is, we have both the inner product and its conjugate in the sum. By elementary properties of complex inner products, from these pairs of terms, the complex parts will vanish and we will be left with two equal real parts. The full calculation of $p_{g_0g_0}$ then goes as the following:
	\begin{align*}
		p_{g_0g_0} = & \frac{p}{3}  \sum_{i, j} \bra{g_0g_0} \kb{ii}{jj} \ket{g_0g_0} + (1 - p) \bra{g_0g_0} \frac{\mathbb{I}_{AB}}{9} \ket{g_0g_0} \\
		&=  \frac{p}{3} \left( \sum_{i} \bra{g_0g_0} \kb{ii}{ii} \ket{g_0g_0} + \sum_{i \ne j} \bra{g_0g_0} \kb{ii}{jj} \ket{g_0g_0}\right) +  \frac{(1 - p)}{9}  \times 1\\	
		&= \frac{p}{3} ( \frac{3}{9} + \bk{g_0}{0}^2\bk{1}{g_0}^2 + \bk{g_0}{0}^2\bk{2}{g_0}^2 + \bk{g_0}{1}^2\bk{0}{g_0}^2 \\
		& \hspace{8mm} + \bk{g_0}{1}^2\bk{2}{g_0}^2 + \bk{g_0}{2}^2\bk{0}{g_0}^2 +\bk{g_0}{2}^2\bk{1}{g_0}^2 ) + \frac{(1 - p)}{9}\\	
		&= \frac{p}{3}\left( \frac{3}{9} +\frac{\bk{1}{g_0}^2}{3} + \frac{\bk{2}{g_0}^2}{3} +  \frac{\bk{g_0}{1}^2}{3} + \bk{g_0}{1}^2 \bk{2}{g_0}^2 +\frac{\bk{g_0}{2}^2}{3} + \bk{g_0}{2}^2\bk{1}{g_0}^2\right) + \frac{(1 - p)}{9}\\
		&= \frac{p}{3}\left( \frac{1}{3} +\frac{\bk{1}{g_0}^2 + \bk{g_0}{1}^2}{3} + \frac{\bk{2}{g_0}^2 + \bk{g_0}{2}^2}{3} +  \bk{g_0}{1}^2 \bk{2}{g_0}^2 + \bk{g_0}{2}^2\bk{1}{g_0}^2 \right) + \frac{(1 - p)}{9}\\	
		&=\frac{p}{3}\left(\frac{1}{3} +\frac{\omega + \omega^2}{9} + \frac{\omega + \omega^2}{9} + \frac{\om^3}{9} + \frac{\om^3}{9} \right) + \frac{(1 - p)}{9}\\		
		&=\frac{p}{3}\left(\frac{1}{3} -\frac{1}{9}-\frac{1}{9} + \frac{1}{9} + \frac{1}{9} \right) + \frac{(1 - p)}{9}\\				
		&=\frac{p}{3} \times \frac{1}{3} + \frac{(1 - p)}{9} \\
		&= \frac{1}{9},
	\end{align*}
	remembering that $1 + \om + \om^2 = 0$ and $\om^3 = 1$. 
	A similar calculation for all $p_{g_ig_j}$, for $i, j \in \{0, 1, 2\}$ reveals that in each case, $p_{g_ig_j} = \frac{1}{9}$. So for measurement base $M_2$, the mutual information 
	\begin{align*}
		I(M^A_2 : M^B_2) = H(M_2^A) + H(M_2^B) - \sum_{i = 1}^{9} \frac{1}{9} \log \frac{1}{9} = 2 \log(3) - 2 \log(3) = 0, 
	\end{align*}
	as $H(M_2^A) = H(M_2^B) = \log(3)$. Performing this same calculation for measurement base $M_3$ also reveals that $I(M^A_3 : M^B_3) = 0$. Remembering inequality \eqref{eq:genmutinf}, considering $d = 3$, we have that the mutual information  $I(A : B)$ in this dimension is
	\begin{align*}
		I(A : B)&=2  \log(3) + \frac{p(3^2 - 1) + 1}{3^2} \log \left( \frac{p(3^2 - 1) + 1}{3^2} \right) + \sum_{i = 1}^{3^2 - 1} \frac{1-p}{3^2} \log \left( \frac{1-p}{3^2} \right) \nonumber \\
		&= 2  \log(3) + \frac{8p + 1}{9} \log \left( \frac{8p + 1}{9} \right) + \sum_{i = 1}^{8} \frac{1-p}{9} \log \left(\frac{1-p}{9}\right)
	\end{align*}
	Now we sum the resulting mutual information after $A$ and $B$ performs all $4$ MUB measurements listed in \eqref{eq:dim3mats}, compare that with $I(A : B)$ and get: 
	\begin{align}
		&I(A : B) \ge I(Z^A : Z^B) + I(X^A : X^B) + I(M^A_2 : M^B_2) + I(M^A_3 : M^B_3)\nonumber\\
		\implies &2  \log(3) + \frac{8p + 1}{9} \log \left( \frac{8p + 1}{9} \right) + \sum_{i = 1}^{8} \frac{1-p}{9} \log \left( \frac{1-p}{9} \right) \nonumber\\
		&\ge 2 \times \left(2 \log(3) +  3 \times \left(\frac{2 p + 1}{9} \right)  \log \left(\frac{2p + 1 }{9} \right) + 6 \times \frac{1 - p}{9} \log \left(\frac{1-p }{9}\right) \right) + 0 + 0. \label{eq:ecqcdim3}
	\end{align}
	We can easily see that this inequality holds for any $p \in [-\frac{1}{d^2 - 1}, 1]$. Hence we can say that, ECQC holds for isotropic states with $3$-dimensional subsystems. In figure \eqref{fig:figure3} we show the graph that depicts this inequality \eqref{eq:ecqcdim3}. 

	\subsubsection{Mutual information in isotropic states for arbitrary prime dimension }
	Generalizing the calculations for isotropic states with $3$-dimensional subsystems to isotropic states with prime dimensional subsystems, we can arrive at the following proposition, which we state without a proof.
	\begin{proposition} \label{prop:measmutinf}
		For isotropic states with subsystems of prime dimension $d$, except for two mutually unbiased basis measurements $\Z$ and $\X$, any other measurement basis $M_i$ will produce zero mutual information, that is, $I(M_i^A : M_i^B) = 0$, and for $\Z$ and $\X$:
		\begin{align*}
			I(Z^A : Z^B) = I(X^A : X^B) = 2 \log(d) + d\left(\frac{p}{d} + \frac{1 - p}{d^2}\right)\log\left(\frac{p}{d} + \frac{1 - p}{d^2}\right) + d^2 \times \frac{1 - p}{d^2} \log\left(\frac{1 - p}{d^2}\right). 
		\end{align*}		
	\end{proposition}	
	\begin{proof}
		We can follow a similar direct calculation strategy that we have followed for dimension $3$. It is easy to see that the arguments for this particular dimension $d = 3$ is still valid for arbitrary prime $d$ and the proposition \eqref{prop:measmutinf} follows. 
	\end{proof}
	
	\noindent We have already mentioned the mutual information for prime dimensional subsystems in inequality \eqref{eq:genmutinf}. Together with proposition \eqref{prop:measmutinf}, we can write the following inequality whose validity will imply the validity of ECQC in isotropic states with arbitrary prime dimensional subsystems:
	\begin{align}
		& I(A : B) \ge \sum_{i = 1}^{d + 1} I(M^A_i : M^B_i) \nonumber\\
		&\implies 2  \log(d) + \frac{p(d^2 - 1) + 1}{d^2} \log \left( \frac{p(d^2 - 1) + 1}{d^2} \right) + \sum_{i = 1}^{d^2 - 1} \frac{1-p}{d^2} \log \frac{1-p}{d^2} \nonumber\\
		&\ge 2 \times \left( 2 \log(d) + d\left(\frac{p}{d} + \frac{1 - p}{d^2}\right)\log\left(\frac{p}{d} + \frac{1 - p}{d^2}\right) + d^2 \times \frac{1 - p}{d^2} \log\left(\frac{1 - p}{d^2}\right) \right) \nonumber\\
		&\ge \underset{\mathcal{S} \in [d + 1]^d}{\min} \sum_{i \in \mathcal{S}} I(M^A_i : M^B_i) \label{eq:extra}. 
	\end{align}	
	In the following graph, we simulate the difference in the left- and right-hand side of the inequality \eqref{eq:extra} for different prime dimensions $d$. Note that, we do not need to exclude the largest classical mutual information in this particular case, as for isotropic states, $I(A:B)$ is actually greater than all $d + 1$ classical mutual information generated by the MUBs in that prime dimension.
	\begin{figure}[H]
		\centering
		\includegraphics[width=.7\linewidth]{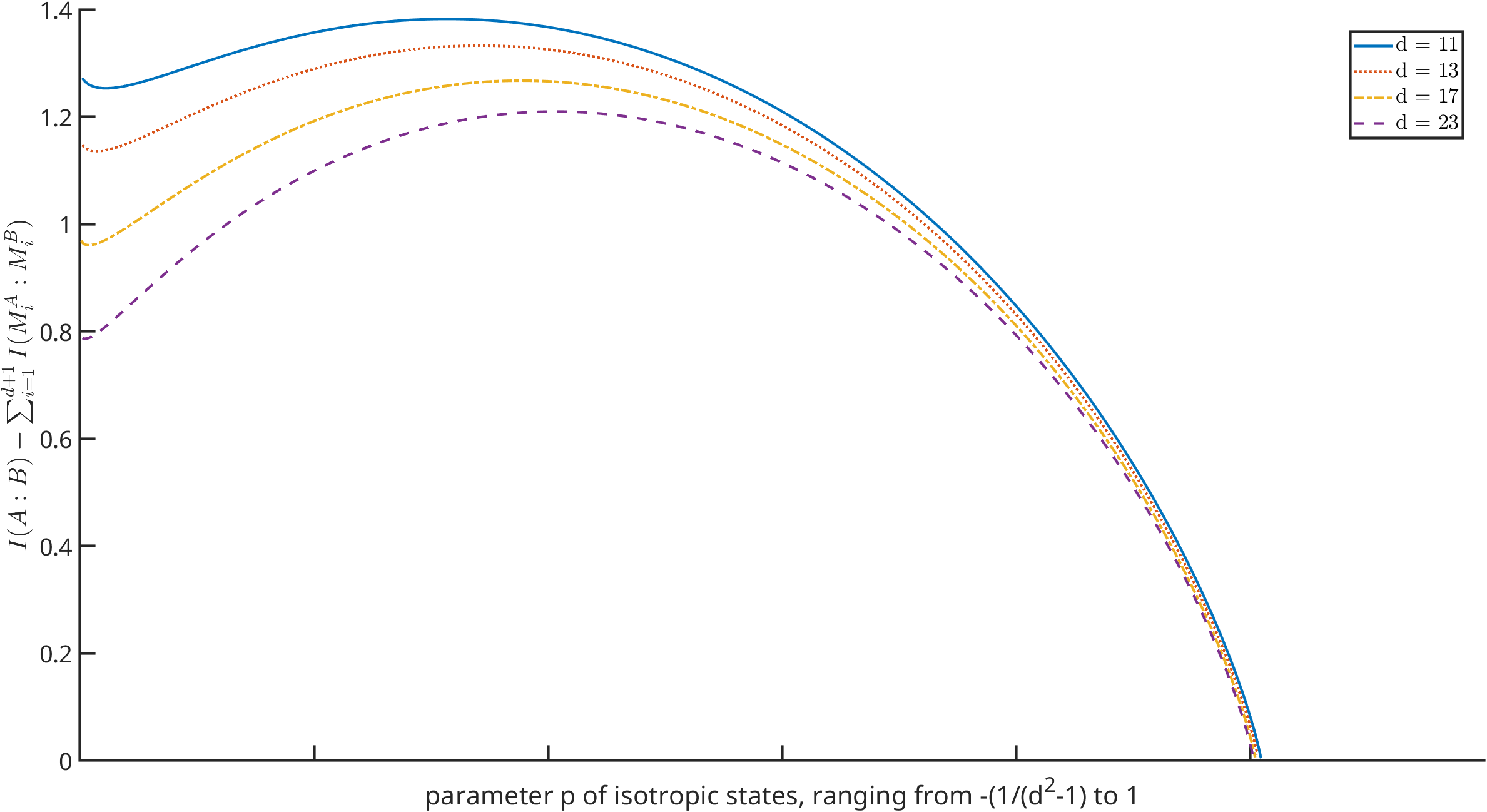}
		\caption{Graph showing the behavior of various prime dimensions ($d = 11, 13, 17, 23$) and the ECQC conjecture applied to the isotropic states. }
		\label{fig:simuldimdisotropic}
	\end{figure}
	\subsection{Simulations of ECQC in Random states}
	\subsubsection{Dimension 3}
	Using the measurement bases mentioned in inequality \eqref{eq:dim3mats}, we perform simulations of our ECQC conjecture on pure bipartite, random bipartite, and isotropic states, where all these states have three dimensional subsystems. We show that taking any three of these measurement bases, thus choosing one less than the maximum number of MUBs in this dimension, does not violate the ECQC conjecture for these states. We present the results of these simulations in figure \eqref{fig:figure3}. In this figure, in (a), where we simulate the ECQC conjecture for 100000 random bipartite pure states, we see that the difference between the left and right hands of the conjecture is never negative. Similarly, in (b), we perform the same simulation for 100000 random bipartite states, observing no negative difference between the left and right hand sides of the ECQC conjecture. In (c), we consider isotropic states with a three dimensional subsystem for the range of parameters $-\frac{1}{d^2 - 1} \le p \le 1$. We notice that, in this dimension $d = 3$, the isotropic states follow a similar trend that we have seen in figure \eqref{fig:simuldimdisotropic} and do not find any contradiction to the ECQC conjecture.  
	\begin{figure}[H]		
		\centering
		\begin{subfigure}[b]{.45\textwidth}
			\includegraphics[width=\textwidth]{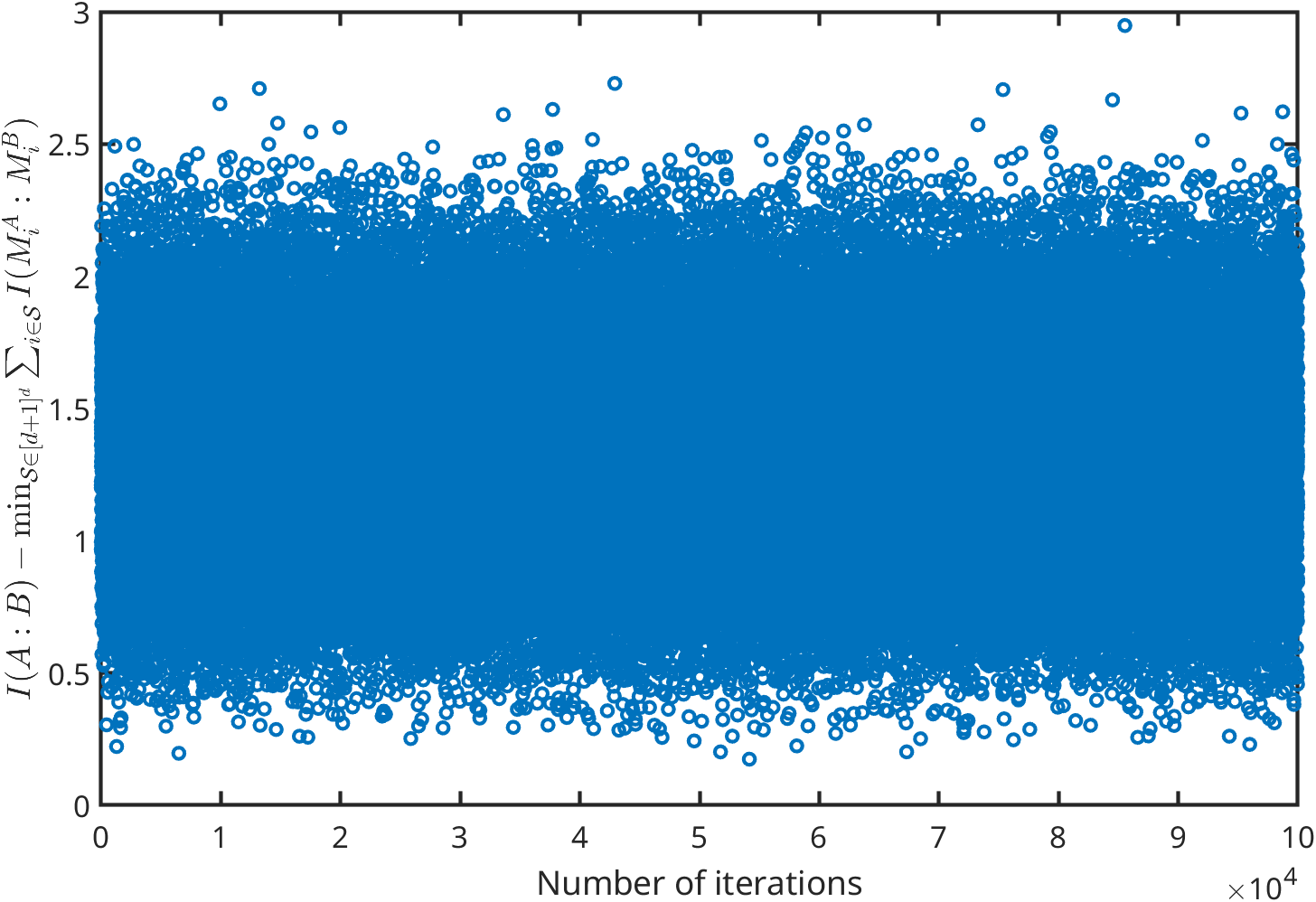}
			\caption{Pure states}
		\end{subfigure}
		\begin{subfigure}[b]{.45\textwidth}
			\includegraphics[width=\textwidth]{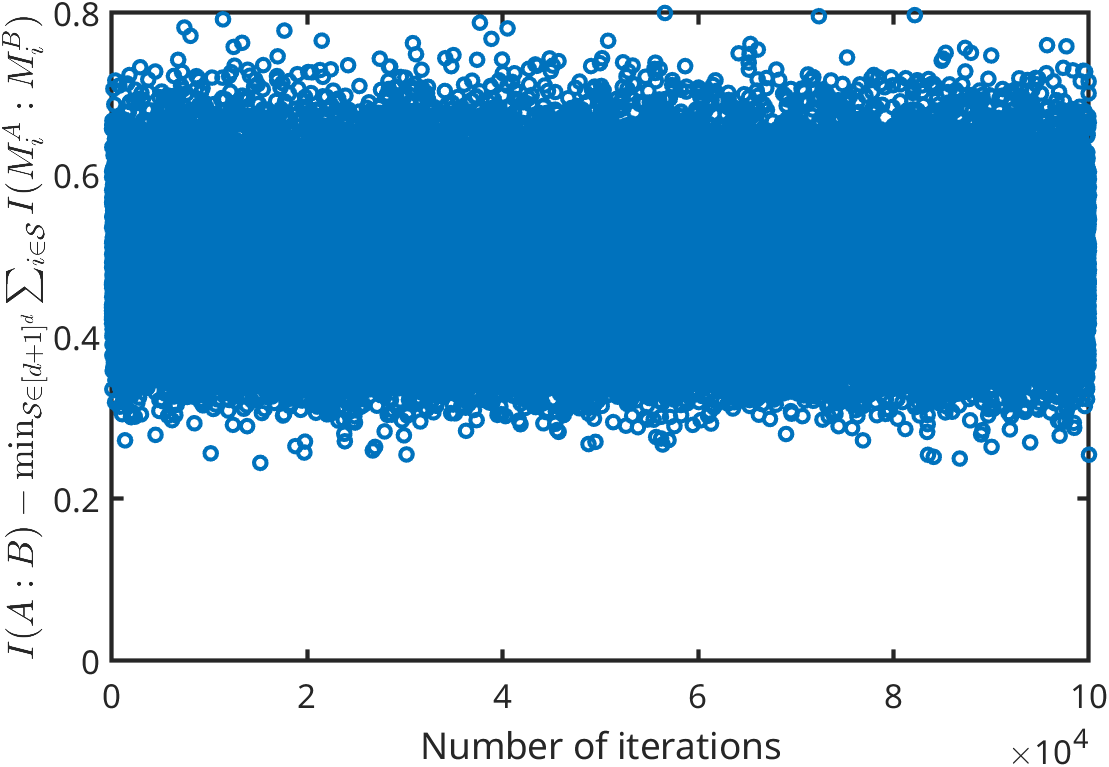}
			\caption{Random states}
		\end{subfigure}
		\begin{subfigure}[b]{.5\textwidth}
			\includegraphics[width=\textwidth]{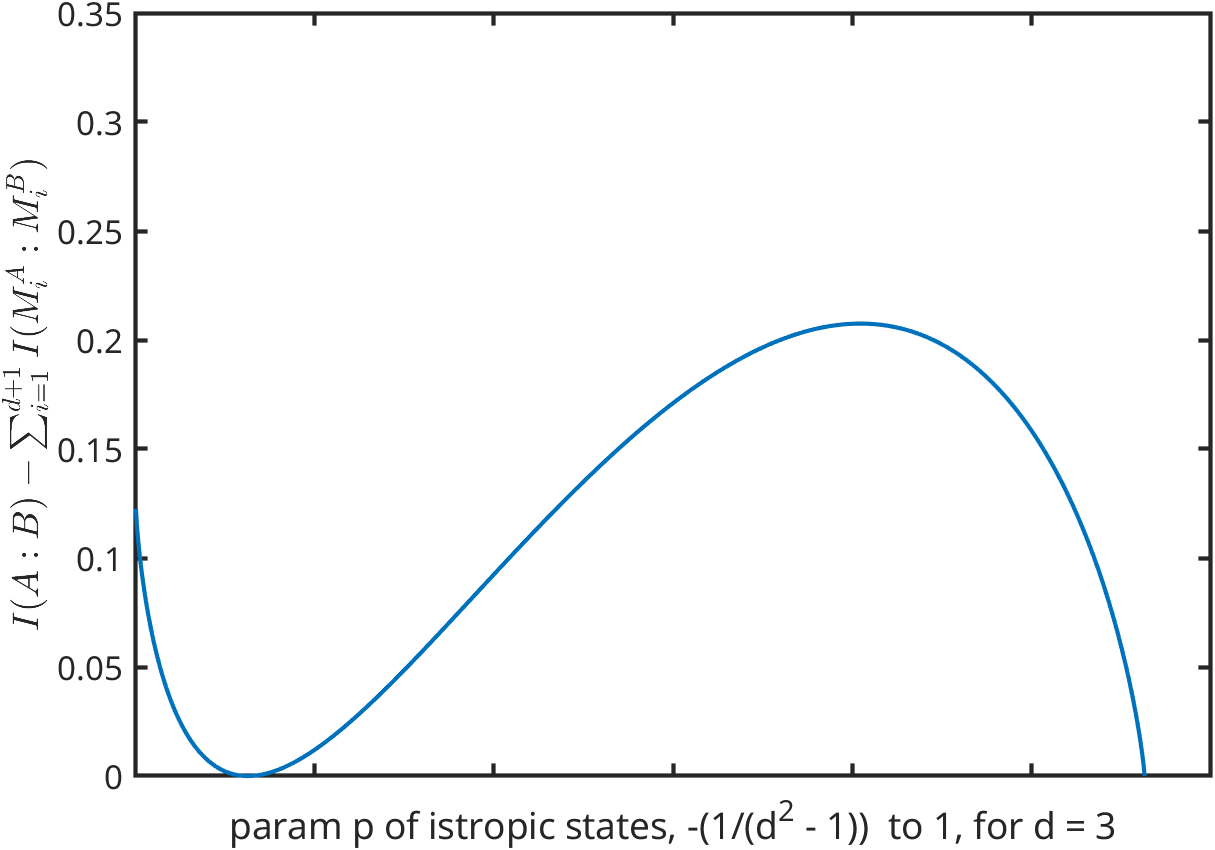}
			\caption{Isotropic states}
		\end{subfigure}		
		\caption{Simulation of ECQC on different types of states with three dimensional subsystems. }
		\label{fig:figure3}		
	\end{figure}
	
	
	\subsubsection{Dimension 5}
	In this section, we simulate the ECQC for dimension $5$ for pure, random and isotropic bipartite states.	We consider the following set of MUBs in this dimension, whose formulation is found in \cite{brierley2009mutually}. These matrices are generated by the following special matrices, where note that, $M_0$ itself is the quantum Fourier matrix of order 5, and is one of the MUB bases in this dimension:
	
	\begin{align*}
		D = 
		\begin{bmatrix}
			1 & 0 & 0 & 0 & 0 \\
			0 & \om & 0 & 0 & 0\\
			0 & 0 & \om^4 & 0 & 0 \\
			0 & 0 & 0 & \om^4 & 0 \\
			0 & 0 & 0 & 0 & \om \\
		\end{bmatrix}, 
		M_0 = \frac{1}{\sqrt{5}}
		\begin{bmatrix}
			1 & 1 & 1 & 1 & 1 \\
			1 & \om & \om^2 & \om^3 & \om^4\\
			1 & \om^2 & \om^4 & \om  & \om^3 \\
			1 & \om^3 & \om & \om^4 & \om^2 \\
			1 & \om^4 & \om^3 & \om^2 & \om \\
		\end{bmatrix}.	
	\end{align*}
	Based on these two matrices $D$ and $M_0$ and remembering that identity matrix, which is denoted by $M_1$ in the following, is also part of the six MUBs found in this dimension, we generate the four remaining MUB bases:
	\begin{align*}
		M_1 = 
		\begin{bmatrix}
			1 & 0 & 0 & 0 & 0 \\
			0 & 1 & 0 & 0 & 0 \\			
			0 & 0 & 1 & 0 & 0 \\
			0 & 0 & 0 & 1 & 0 \\
			0 & 0 & 0 & 0 & 1		    		    		    
		\end{bmatrix},
		M_2 = DM_0 = \frac{1}{\sqrt{5}}
		\begin{bmatrix}
			1 & 1 & 1 & 1 & 1 \\
			\om & \om^2 & \om^3 & \om^4 & 1\\
			\om^4 & \om & \om^3 & 1  & \om^2 \\
			\om^4 & \om^2 & 1 & \om^3 & \om \\
			\om & 1 & \om^4 & \om^3 & \om^2 \\
		\end{bmatrix}, \\
		M_3 = D^2M_0 = \frac{1}{\sqrt{5}}
		\begin{bmatrix}
			1 & 1 & 1 & 1 & 1 \\
			\om^2 & \om^3 & \om^4 & 1 & \om\\
			\om^3 & 1 & \om^2 & \om^4  & \om \\
			\om^3 & \om & \om^4 & \om^2 & 1\\
			\om^2 & \om & 1 & \om^4 & \om^3 \\
		\end{bmatrix}, 
		M_4 = D^3M_0 = \frac{1}{\sqrt{5}}
		\begin{bmatrix}
			1 & 1 & 1 & 1 & 1 \\
			\om^3 & \om^4 & 1 & \om & \om^2\\
			\om^2 & \om^4 & \om & \om^3  & 1 \\
			\om^2 & 1 & \om^3 & \om & \om^4\\
			\om^3 & \om^2 & \om & 1 & \om^4 \\
		\end{bmatrix}, 
		\\
		M_5 = D^4M_0 = \frac{1}{\sqrt{5}}
		\begin{bmatrix}
			1 & 1 & 1 & 1 & 1 \\
			\om^4 & 1 & \om & \om^2 & \om^3\\
			\om & \om^3 & 1 & \om^2  & \om^4 \\
			\om & \om^4 & \om^2 & 1 & \om^3\\
			\om^4 & \om^3 & \om^2 & \om & 1 \\
		\end{bmatrix}, 
	\end{align*}
	where $\om = e^{\frac{2 \pi i}{5}}$ is the fifth root of unity. Experimenting with these 6 MUB bases, we observe similar trends for the ECQC conjecture, which we present in figure \eqref{fig:figure5}. In this figure, in (a) we take 10000 random pure states and measure them in these $6$ MUB bases. We are considering less number of states compared to dimension 3 because of the computational complexity for high-dimensional states. Then we take out the largest mutual information produced by any of these bases and sum up the rest of the $5$ resulting mutual information. Then we compare this sum with the original mutual information that was present in those states before the measurement and observe that the original mutual information is always at least large as this particular sum, which signifies the validity of the ECQC conjecture in this dimension. Similarly, in (b) we consider $10000$ random bipartite states and perform this simulation, observing no contradiction with the ECQC conjecture. In (c), we measure isotropic states with five dimensional subsystems to see the validity of this conjecture. It turns out in this prime dimension $d = 5$, only bases $M_0$ and $M_1$ produce non-zero mutual information after the bipartite measurement in these bases. The post-measurement state after the measurement by $M_2, M_3, M_4, M_5$ does not have any mutual information, which actually motivated our proposition \eqref{prop:measmutinf} for arbitrary prime dimensions for this class of states. So in (c), it is noticeable that on the Y-axis, we are not considering the sum of mutual information of the post-measurement states excluding the largest one. We are actually considering the sum of all six mutual information quantities obtainable from the post-measurement states from these six MUB bases applied on isotropic states. So, this simulation in (c) actually shows a stronger quantity than the ECQC conjecture, further providing credence to it.  
	\begin{figure}[H]		
		\centering
		\begin{subfigure}[b]{.45\textwidth}
			\includegraphics[width=\textwidth]{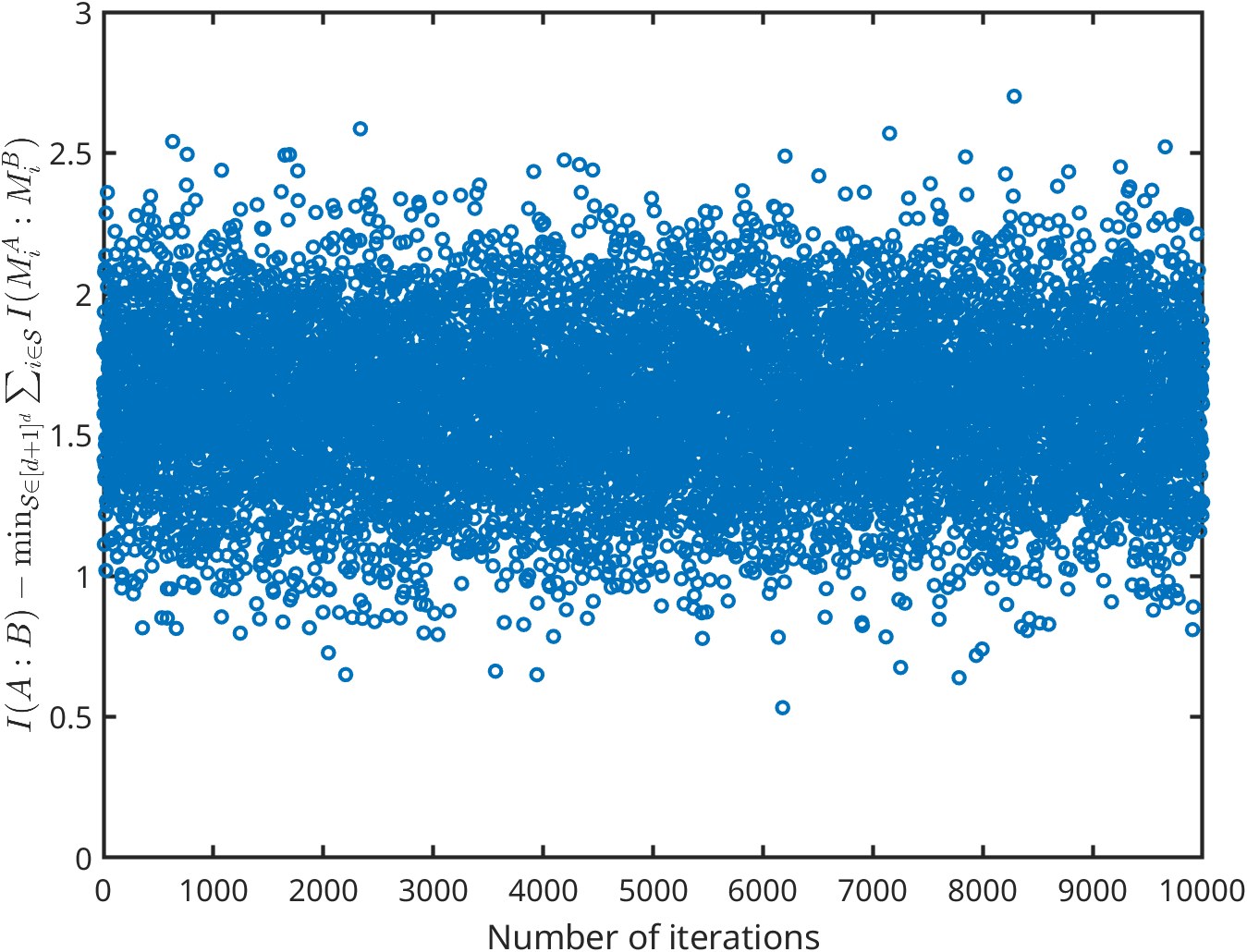}
			\caption{Pure states}
		\end{subfigure}
		\begin{subfigure}[b]{.45\textwidth}
			\includegraphics[width=\textwidth]{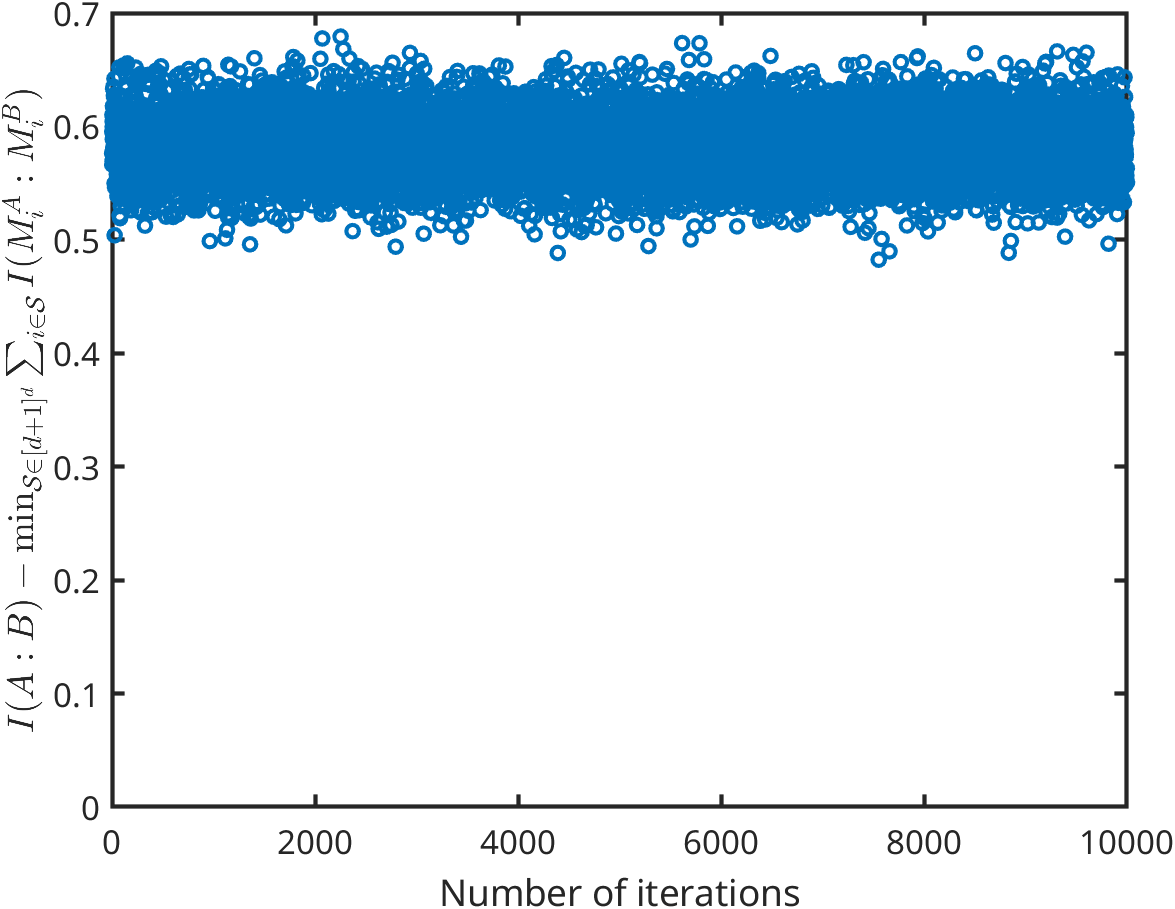}
			\caption{Random states}
		\end{subfigure}
		\begin{subfigure}[b]{.45\textwidth}
			\includegraphics[width=\textwidth]{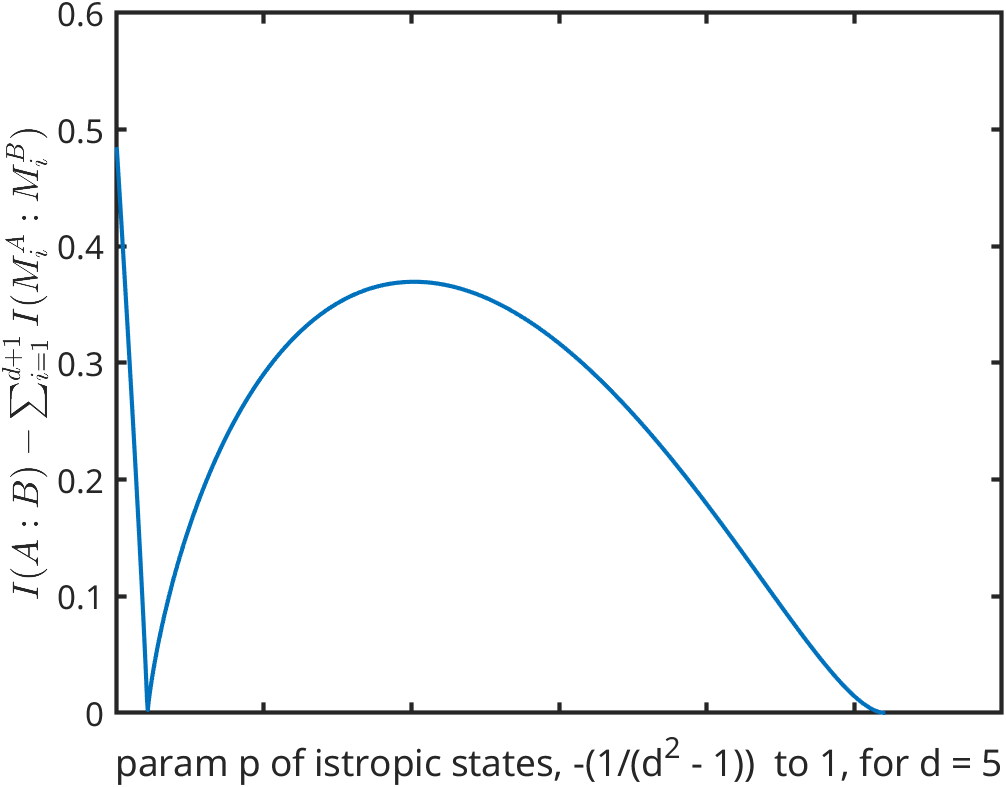}
			\caption{Isotropic states}
		\end{subfigure}		
		\caption{Simulation of ECQC on different types of states with five dimensional subsystems. }
		\label{fig:figure5}		
	\end{figure}
	
	\section{Closing Remarks and Future Works}
	In this work, we have taken a deeper look at the CQC conjecture and identified a potentially useful sufficient condition. We have also proposed an extended version of this conjecture that takes into account high-dimensional states and more measurement bases. There are a few interesting works that could be done in the future. For example, the original CQC conjecture holds for all pure states, but the technique that was used to prove this turned out to be inapplicable for the ECQC conjecture. However, as our simulations for dimensions $3$ and $5$ show, there are no observed contradictions of the ECQC for pure states. So it would be interesting to develop a method to prove ECQC for pure states analytically, similar to the case of the original CQC conjecture. It would also be interesting to formulate a conjecture that holds for both prime and composite dimensions. The difficulty in such a formulation is that the number of MUBs in a composite dimension is not straightforward to calculate, and using available MUBs sometimes leads to counterexamples of ECQC. So, for composite dimensions, we may need an entirely new formulation of the CQC conjecture, that holds for higher dimensions and more than two measurement bases. 
	
	\section*{Acknowledgment}
	Part of this work was performed when the author was at the University of Connecticut, under the supervision of Dr. Walter O. Krawec.
	\bibliography{references}
	\bibliographystyle{unsrt}
	
\end{document}